\documentclass[11pt,english]{article}
\usepackage{color}
\usepackage[T1]{fontenc}
\usepackage[latin9]{inputenc}
\usepackage[a4paper]{geometry}
\usepackage{textcomp}
\usepackage{amsthm}
\usepackage{amsmath}
\usepackage{amssymb}
\usepackage{esint}
\usepackage{bbm}
\usepackage[colorlinks=true]{hyperref}
\usepackage{babel}
\usepackage{xcolor}

\usepackage{amscd}

\usepackage{amsfonts}
\usepackage{amsthm}
\usepackage{mathrsfs}
\usepackage{dsfont}

\usepackage{mathtools}

\usepackage{enumerate}
\usepackage{bm}
\usepackage{bbm}
\usepackage{verbatim}

\usepackage{enumitem}
\newlist{abbrv}{itemize}{1}
\setlist[abbrv,1]{label=,labelwidth=1.2in,align=parleft,itemsep=0.1\baselineskip,leftmargin=!}

\numberwithin{equation}{section}

\makeatletter

\DeclareFontEncoding{LGR}{}{}
\DeclareTextSymbol{\~}{LGR}{126}

\newcommand{\tr}{\text{Tr}}

\newcommand{\abs}[1]{\left| #1 \right|}

\newcommand{\scp}[2]{\big\langle #1 , #2 \big\rangle}

\newcommand{\SCP}[2]{\big\langle #1 , #2 \big\rangle}

\newcommand{\bra}[1]{\langle #1 |}

\newcommand{\ket}[1]{| #1 \rangle}

\newcommand{\norm}[1]{\left\| #1 \right\| }

\renewcommand{\Re}{\mathrm{Re}}
\renewcommand{\Im}{\mathrm{Im}}

\newcommand{\id}{\mathds{1}}

\newcommand{\op}{\mathrm{op}}

\newcommand{\be}{\begin{equation}}
\newcommand{\ee}{\end{equation}}

\newtheorem{theorem}{Theorem}[section]
\newtheorem{lemma}{Lemma}[section]
\newtheorem{def:lemma}{Definition and Lemma}[section]

\newtheorem{proposition}{Proposition}[section]

\theoremstyle{definition}
\newtheorem{definition}{Definition}[section]

\theoremstyle{remark}
\newtheorem{remark}{Remark}[section]

\usepackage[colorinlistoftodos]{todonotes}

\usepackage{ulem}

\allowdisplaybreaks[1]

\begin{document}


\title{Derivation of the Landau--Pekar equations in a many-body mean-field limit }

\author{Nikolai Leopold\footnote{Corresponding author, University of Basel, Department of Mathematics and Computer Science, 
Spiegelgasse 1, 4051 Basel, Switzerland,
E-mail address: \texttt{nikolai.leopold@unibas.ch} }, David Mitrouskas\footnote{Universit\"at Stuttgart, Fachbereich Mathematik, Pfaffenwaldring 57,
70569 Stuttgart, Germany. Current address: Institute of Science and Technology Austria (IST Austria), Am Campus 1, 3400 Klosterneuburg, Austria,
E-mail address: \texttt{dmitrous@ist.ac.at}}, and Robert Seiringer\footnote{Institute of Science and Technology Austria (IST Austria), Am Campus 1, 3400 Klosterneuburg, Austria,
E-mail address: \texttt{robert.seiringer@ist.ac.at}}
}

\maketitle

\frenchspacing

\begin{abstract}
We consider the Fr\"ohlich Hamiltonian in a mean-field limit where many bosonic particles weakly couple to the quantized phonon field.
For large particle number and suitably small coupling, we show that the dynamics of the system is approximately described by the Landau--Pekar equations. These describe a Bose--Einstein condensate interacting with a classical polarization field, whose dynamics is effected by the condensate, i.e., the back-reaction of the phonons that are created by the particles during the time evolution  is of leading order.

\end{abstract}

\section{Introduction}

We consider the dynamics of $N$ bosonic particles interacting with a quantized phonon field described by the Fr\"ohlich model in a mean field regime. The underlying Hilbert space is
\begin{align}
\mathcal{H}^{(N)} = L^2_s\left( \mathbb{R}^{3N} \right) \otimes \mathcal{F}_s ,
\end{align}
where the $N$ particles are described by states in $L^2_s(\mathbb R^{3N})$, the subspace of all complex-valued square integrable $N$-particle wave functions that are symmetric under the exchange of any pair of the coordinates $(x_1,...,x_N)$, and where the phonon field is represented by elements in the bosonic Fock space 
$
 \mathcal{F}_s  =  \bigoplus_{n \geq 0} L^2_s(\mathbb{R}^{3n}).
$
The time evolution of the system is governed by the Schr\"odinger equation
\begin{align}
\label{eq: Schroedinger equation}
 i \partial_t \Psi_{N,t} = H_{N,\alpha}^{\rm F} \Psi_{N,t}
\end{align}
with Fr\"ohlich Hamiltonian
\begin{align}
\label{eq: Froehlich Hamiltonian quadratic form}
H_{N,\alpha}^{\rm F} &= \sum_{j=1}^N  \left[ - \Delta_j  +  \sqrt{\frac{\alpha}{N}}  \int d^3k \, \abs{k}^{-1}
\left( e^{ikx_j} a_k + e^{-ikx_j } a^*_k  \right)  \right] + \mathcal{N} .
\end{align}
Here, $\Delta_j$ is the Laplacian acting on the $j$th particle with coordinate $x_j$, $a_k$ and $a_k^*$ denote the usual bosonic annihilation and creation operators satisfying the canonical commutation relations
\begin{align}
\label{eq: canonical commutation relation}
[a_k, a^*_l ] &=  \delta(k-l), \quad
[a_k, a_l ] = 
[a^*_k, a^*_l ] = 0,
\end{align}
and $\mathcal N$ is the number operator defined by $\mathcal N= \int d^3 k\, a_k^* a_k$. The coupling parameter $\sqrt{\alpha/N}$ is introduced to scale the strength of the interaction between the particles and the phonon field. If the number of phonons is of order $N$ and $\alpha >0$ is fixed, the factor $N^{-1/2}$ and the fact that the creation and annihilation operators scale like $\sqrt{N}$ (they are  bounded by $\left(\mathcal{N} + 1 \right)^{1/2}$, see \eqref{eq: bound for annihialtion and creation operators}) ensure that the kinetic and potential energy are of the same order for large $N$.

We note that the expression \eqref{eq: Froehlich Hamiltonian quadratic form} is somewhat formal, since the form factor $\vert k\vert^{-1}$ in the interaction term is not square integrable. By a well-known argument going back to Lieb and Yamazaki \cite{liebyamazaki} (cf. Lemma \ref{lemma: bounds for the interaction term}), the right side of \eqref{eq: Froehlich Hamiltonian quadratic form} defines a closed bounded from below quadratic form with domain given by the form domain of $H_{N,0}^{\rm F}$. The self-adjoint operator that corresponds to this form is called Fr\"ohlich Hamiltonian and denoted by $H_{N,\alpha}^{\rm F}$. We refer to \cite{griesemerwuensch} for a detailed description of its domain $\mathcal D(H_{N,\alpha}^{\rm F})$ (see also Lemma~\ref{lemma representation of H_F}). 

\allowdisplaybreaks

If the number $N$ of particles is large, we show for a particular class of initial states that the solution of the many-body Schr\"odinger equation \eqref{eq: Schroedinger equation} can be approximated by Pekar product states, i.e., states of the form
\begin{align}\label{eq: Pekar-State}
\Psi_{N,t} = \psi_t^{\otimes N}\otimes W(\sqrt N \varphi_t) \Omega,
\end{align}
where $\Omega$ is the vacuum state in $\mathcal F_s$, $W$ denotes the Weyl operator  and $(\psi_t,\varphi_t) \in L^2(\mathbb R^3) \times L^2(\mathbb R^3)$ solve the time-dependent Landau--Pekar equations
\begin{align}
\label{eq: Landau Pekar equations}
\begin{cases}
 i \partial_t \psi_t(x) & = \ \ \left[ - \Delta_x + \sqrt{ \alpha}  \Phi(x,t)  \right] \psi_t(x),  \\[2mm]
i \partial_t \varphi_t(k) & = \ \   \varphi_t(k) + \sqrt{ \alpha } \abs{k}^{-1}  \int d^3x \,  e^{-ikx} \abs{\psi_t(x)}^2  
\end{cases}
\end{align}
where 
\begin{align}
\Phi(x,t) &=
 \int d^3k \, \abs{k}^{-1}
\left(  e^{ikx} \varphi_t(k)  + e^{-ikx}  \overline{\varphi_t(k)}  \right).
\end{align}
The Weyl operator is defined for any $f\in L^2(\mathbb R^3)$ by
\begin{align}
\label{eq: Weyl operator}
W(f) = \exp \left(  \int d^3k \, \big( f(k) a^*_k - \overline{f(k)} a_k \big) \right).
\end{align} 
In the Pekar product state \eqref{eq: Pekar-State}, the phonons are in the coherent state $W (\sqrt N \varphi_t) \Omega$ with average number of excitations of order $N$, and the bosonic particles form a pure Bose--Einstein condensate with condensate wave function $\psi_t$. According to the Landau--Pekar equations, the one-particle condensate wave function $\psi_t$ evolves in the potential $\sqrt\alpha \Phi(x,t)$ created by the phonons, while the phonon field couples to the particles via the source term involving the  density $\abs{\psi_t(x)}^2$.

Our main result can be summarized as follows: Given an initial wave function $\Psi_{N,0}$ that is close to a Pekar product state $\psi_0^{\otimes N}\otimes W(\sqrt N \varphi_0)\Omega$ (close in an appropriate sense that will be specified in the next section), then the time evolved state $e^{-iH_{N,\alpha}^{\rm F} t}\Psi_{N,0}$ remains close to the time evolved Pekar state \eqref{eq: Pekar-State} when $N\gg 1$.

The Landau--Pekar equations were originally introduced in \cite{landaupekar} to approximate the time evolution of a single polaron in the strong coupling limit. In our notation, the strong coupling regime corresponds to the Hamiltonian $H_{1,\alpha}^{\rm F}$ with $\alpha \gg 1$. Partial results concerning a rigorous derivation of the Landau--Pekar equations in the strong coupling limit were obtained in \cite{frankschlein, frankgang, griesemer, LRSS} (for a detailed comparison between the different results we refer to \cite[Chapter 2]{LRSS}). In these works, the Landau--Pekar equations are justified for short times, namely at most for times of order $\alpha^{-\varepsilon}$ with $\varepsilon >0$ arbitrary small.\footnote{It should be noted that results about the polaron in the strong coupling limit are usually formulated in strong coupling units and that times of order $\alpha^2$ in the stated references correspond to times of order one in the units of the present paper.} A derivation for times of order one, the time scale in the strong coupling limit at which the back-reaction of the phonons that are created during the time evolution is of leading order, remains an open problem. The emergence of classical radiation in the strong coupling limit is expected to rely on the adiabatic decoupling between the relatively fast moving (w.r.t. $\alpha$) electron and the radiation field. For results on adiabatic theorems of the Landau--Pekar equations in one and three dimensions we refer to \cite{frankgang2} and \cite{LRSS}.

In the many-particle mean-field limit considered in this work, the creation of coherent radiation happens for a different reason than in the strong coupling regime, namely because there are many particles in the same quantum state that simultaneously create the phonons. In this regard, the present work is related to \cite{ammarifalconi, falconi,  leopoldpetrat, leopoldpickl, leopoldpickl2}, where  many-body mean-field limits of the renormalized Nelson model, the Nelson model with ultraviolet cutoff and the (bosonic) Pauli--Fierz model are considered. In particular, we mention \cite{ammarifalconi} where the Schr\"odinger--Klein--Gordon equations were derived by the Wigner measure approach as a limit of the renormalized Nelson model. 

In \cite{ginibrenironivelo, CCFO, CFO}, effective equations for the Nelson, Pauli--Fierz and Fr\"ohlich model were derived in a partially classical limit. There, the number of particles is kept fixed while the number of excitations of the quantum field tends to infinity and the coupling constant approaches zero in a suitable sense. The effect of the excitations that are created during time evolution is negligible in this limit and the quantum field can thus be approximated by a classical field that evolves freely or remains constant in time.

To the  best of our knowledge, the present work provides the first derivation of the Landau--Pekar equations in a limit in which the back-reaction of the phonons that are created during time evolution is of leading order. Moreover, our results include explicit error estimates. 

In order to derive our results, we follow \cite{leopoldpickl2}, which combines the methods from \cite{pickl1} and \cite{rodnianskischlein}. The new technical challenge in comparison with \cite{leopoldpickl2} is to show that the high momentum phonons do not obstruct the expected mean-field behavior. This requires several nontrivial modifications. First, it is crucial to introduce a measure for the excitations around the condensate resp.\ around the coherent state that involves the canonical transformation due to Gross and Nelson (see \eqref{eq: definition Gross transform}). In particular, we use the representation of the Fr\"ohlich Hamiltonian in \cite{griesemerwuensch}. The most difficult part is to control the interaction between the ultraviolet modes of the phonon field and the fraction of particles not in the condensate. To this end, we restrict our consideration to a subclass of the initial states which have small fluctuations in the energy per particle observable and combine estimates similar to \cite[Sect. VIII.1]{leopoldpickl} with an operator bound that is motivated by \cite[Lemma 10]{frankschlein}. The idea of using this restriction in order to treat the singular interaction between quantum fields and particles in the mean field regime was already used in \cite{leopoldpickl}.

The article is organized as follows. In the next section, we state our main results. In Theorem \ref{theorem: main theorem}, we consider initial states in the domain of the Fr\"ohlich Hamiltonian, while Theorem \ref{theorem: main theorem: 2} is about initial states in the domain of the noninteracting model (including, in particular, product states). In Section \ref{sec: Notation and basic estimates}, we introduce useful notation and discuss the representation of the Fr\"ohlich Hamiltonian via the Gross transformation. The key steps of the proof of our main result are summarized in Section \ref{section: Proof of the Main Theorem} in terms of several lemmas. The proofs of these are given in Sections \ref{section: variance of the energy bound}--\ref{sec: remaining proofs}.

\section{Main results}

For notational convenience, we set the coupling constant $\alpha = 1$ from now on and denote $H_N^{\rm F} = H_{N,1}^{\rm F}$. All statements and proofs that follow are, however, equally true for any $\alpha >0$ independent of $N$.

In order to state our main results we define for $\Psi_N\in \mathcal{H}^{(N)}$ the one-particle reduced density matrix
\begin{align}
\label{eq: definition reduced one-particle matrix charged particles}
\gamma_{\Psi_N}^{(1,0)} = \tr_{2,\ldots, N} \otimes \tr_{\mathcal{F}_s} \ket{\Psi_{N}} \bra{\Psi_{N}} 
\end{align}
 on the Hilbert space $L^2(\mathbb{R}^3)$. 
Here, $\tr_{2,\ldots, N}$  denotes the partial trace over the coordinates $x_2,\ldots, x_N$ and $\tr_{\mathcal{F}_s}$ the trace over Fock space. The particles of a many-body state $\Psi_{N}$ are said to exhibit complete  Bose--Einstein condensation if there exists 
$\psi \in L^2(\mathbb{R}^3)$ with $\norm{\psi}_{L^2(\mathbb{R}^3)}=1$ such that
\begin{align}
\label{eq: convergence reduced one-particle matrix charged particles}
\tr_{L^2(\mathbb{R}^3)} \abs{\gamma_{\Psi_N}^{(1,0)} - \ket{\psi} \bra{\psi}}  \rightarrow 0
\end{align}
as $N \rightarrow \infty$. In this case $\psi$ is called the condensate wave function. Moreover, we define (for $\left(\psi , \varphi \right) \in L^2(\mathbb{R}^3) \times  L^2(\mathbb{R}^3)$ and $\Psi_{N}\in \mathcal{D} \left( H_N^{\rm F} \right)$) 
\begin{align}
a(\Psi_{N},\psi) & =  \textnormal{Tr}_{L^2(\mathbb{R}^3)}   \abs{\gamma^{(1,0)}_{\Psi_{N }} - \ket{\psi } \bra{\psi}}  , \label{DEF:AN}\\[1mm]
b(\Psi_{N},\varphi) &= N^{-1} \scp{W^* (\sqrt{N} \varphi) \Psi_{N }}{\mathcal{N} \, W^* (\sqrt{N} \varphi) \Psi_{N}}  , \label{DEF:BN} \\[1mm]
c(\Psi_{N}) &=   \norm{N^{-1} \left( H_N^{\rm F} - \scp{\Psi_{N}}{ H_N^{\rm F} \Psi_{N}} \right) \Psi_{N}}^2 .  \label{DEF:CN}
\end{align} 

For $m \in \mathbb{N}$, let $H^m(\mathbb{R}^3)$ denote the Sobolev space of order $m$ and $L_m^2(\mathbb{R}^3)$ a weighted $L^2$-space with norm $\| \varphi\| _{L_m^2(\mathbb{R}^3)} = \| ( 1 + \abs{\,\cdot\,}^2 )^{m/2} \varphi \|_{L^2(\mathbb{R}^3)}$.
We will use the following result which was proven in \cite{frankgang}.

\begin{proposition}[Lemma~C.2 in \cite{frankgang}]\label{prop: solution theory for LPeq}
The Landau--Pekar equations \eqref{eq: Landau Pekar equations} are globally well-posed in $ H^2(\mathbb{R}^3) \times L_1^2(\mathbb{R}^3)$. For all $t \in \mathbb{R}$ we have
\begin{align}
\norm{\psi_t}_{H^2(\mathbb{R}^3)} &\leq C \left( 1 +  \abs{t} \right) 
\quad \text{and} \quad
\norm{\varphi_t}_{L_1^2(\mathbb{R}^3)} \leq C \left( 1 +  \abs{t} \right) \label{BOUNDS:LP:SOLUTIONS}
\end{align}
where $C$ is a constant depending only on the initial data.
\end{proposition} 

We are now ready to state our main theorem.

\begin{theorem}
\label{theorem: main theorem}
Let $(\psi, \varphi) \in H^2(\mathbb{R}^3) \times L_1^2(\mathbb{R}^3)$ s.t. $\norm{\psi}_{L^2(\mathbb{R}^3)} = 1$, and $\Psi_{N} \in \mathcal{D} ( H_N^{\rm F} )$ s.t. $\norm{\Psi_{N}} = 1$ and $E_0 = \sup_{N \in \mathbb{N}} \big\vert N^{-1} \scp{\Psi_{N }}{H_N^{\rm F} \Psi_{N }} \big\vert   <\infty $.
Let $(\psi_t, \varphi_t)$ be the unique solution of \eqref{eq: Landau Pekar equations} with initial datum $(\psi, \varphi)$ and $\Psi_{N,t}= e^{-iH_N^{\rm F} t}\Psi_{N}$. Then, there exists a constant $C >0$ (depending only on $\norm{\varphi }_{L_1^2(\mathbb{R}^3)}$, $\norm{\psi }_{H^2(\mathbb{R}^3)}$ and $E_0$) such that
\begin{align}
\textnormal{Tr}_{L^2(\mathbb{R}^3)}   \abs{\gamma^{(1,0)}_{\Psi_{N,t}} - \ket{\psi_t} \bra{\psi_t}} & \leq \sqrt{ a(\Psi_{N},\psi)  + b(\Psi_{N},\varphi)  + c(\Psi_{N})  +  N^{-1/2}} e^{C (1+\abs{t})^3 } ,\label{MAIN:BOUND:TRACE:NORM:1:0}
\\[3mm]
&\hspace{-4cm} N^{-1} \scp{W^* (\sqrt{N} \varphi_t) \Psi_{N,t}}{\mathcal{N}  \, W^* (\sqrt{N} \varphi_t) \Psi_{N,t}} \label{MAIN:BOUND:TRACE:NORM:0:1} \\[1mm]
 & \leq \left( a(\Psi_{N},\psi)  + b (\Psi_{N },\varphi )  + c(\Psi_{N }) +  N^{-1/2} \right) e^{C ( 1 + \abs{t})^3 } .\nonumber
\end{align}
\end{theorem}

The proof is given in Section \ref{section: Proof of the Main Theorem}.

\begin{remark}
If one considers initial many-body states in which the particles are in a Bose--Einstein condensate, the phonons are in a coherent states and the energy has small fluctuations around its mean value, i.e.
\begin{align}\label{eq: a+b+c to 0}
\lim_{N \rightarrow \infty} \left( a(\Psi_{N},\psi) + b(\Psi_{N},\varphi) + c(\Psi_{N}) \right) = 0
\end{align}
it follows from Theorem \ref{theorem: main theorem} that
\begin{align}
\lim_{N \rightarrow \infty} \textnormal{Tr}_{L^2(\mathbb{R}^3)} \abs{\gamma^{(1,0)}_{\Psi_{N,t}} - \ket{\psi_t} \bra{\psi_t}} 
&=0  \quad \text{and}
\nonumber  \\
\lim_{N \rightarrow \infty}
N^{-1} \scp{W^* (\sqrt{N} \varphi_t) \Psi_{N,t}}{\mathcal{N} \, W^* (\sqrt{N} \varphi_t) \Psi_{N,t}} &=0 .
\end{align}
Our result consequently shows the stability of the condensate and the coherent state during the time evolution.
\end{remark}

\begin{remark}
The condition $c(\Psi_{N}) \rightarrow 0$ as $N \rightarrow \infty$ restricts the initial data to many-body states $\Psi_{N}$ whose energy per particle has small fluctuations around its mean value. In our proof, this is important to obtain sufficient control on the singular ultraviolet behavior of the interaction term in $H_N^{\rm F}$. We give a detailed explanation of this point in Section \ref{section: variance of the energy bound}. In the presence of an ultraviolet cutoff in the Fr\"ohlich Hamiltonian, the estimates \eqref{MAIN:BOUND:TRACE:NORM:1:0} and \eqref{MAIN:BOUND:TRACE:NORM:0:1} hold without the appearance of $c(\Psi_{N})$ on the right hand side, but with a cutoff dependent constant $C$. In this simpler case, the statement could be proven in close analogy to \cite{falconi, leopoldpickl2} where the Nelson model was considered with ultraviolet cutoff.
\end{remark}

Next, we give examples of initial states that satisfy \eqref{eq: a+b+c to 0}. The quantities $a(\Psi_N,\psi)$ and $b(\Psi_N,\varphi)$ are identically zero for Pekar product states $\Psi_{N}=\psi^{\otimes N} \otimes W(\sqrt{N} \varphi) \Omega$ with $(\psi, \varphi) \in H^2(\mathbb{R}^3) \times L_1^2(\mathbb{R}^3)$. However, such Pekar states are in the domain $\mathcal D(H_N^0) = \big( H_s^2(\mathbb{R}^{3N}) \otimes \mathcal{F}_s  \big) \cap \mathcal{D} (\mathcal{N} )$ of the free Hamiltonian
\begin{align}\label{def: free hamiltonian}
H_N^0 = -\sum_{j=1}^N \Delta_{j} + \mathcal N,
\end{align}
and thus, as shown in \cite{griesemerwuensch}, can not be elements of $\mathcal{D} ( H_N^{\rm F} )$. As a consequence, $c(\Psi_{N})$ would be infinite in this case. 
To specify states that satisfy \eqref{eq: a+b+c to 0}, we introduce the Gross transform
\begin{align}
\label{eq: definition Gross transform}
U_{K} = \exp \Bigg[    N^{-1/2} \sum_{j=1}^N \int d^3k \, 
\left(  \overline{B_{K,x_j}(k)} a_k   - B_{K,x_j}(k) a_k^* \right)   \Bigg],
\end{align}
where
\begin{align}
\label{eq: definition B-K-N}
B_{K,x }(k)  &= \frac{-1}{\abs{k} (1 + k^2)} e^{-ikx }   \mathds{1}_{\abs{k} \geq K}(k)  
\end{align}
for $0 < K< \infty$. The Gross transform, which goes back to Gross and Nelson \cite{gross,nelson}, relates the domains of $H_N^0$ and $H_N^{\rm F}$ to each other.\footnote{The Gross transform adds correlations between the bosons and phonon modes with momentum $\vert k \vert \ge  K$. This leads to a better ultraviolet behavior of the radiation field.}
In Lemma \ref{lemma representation of H_F} we show that there is a $\widetilde K>0$ such that for all $K\ge \widetilde K$ and all $N\ge 1$, the domains satisfy
\begin{align}\label{DOMAIN:H:NF}
\mathcal{D} \left( H^{\rm F}_N \right) &= U_K^* \mathcal{D} \left( H_N^0 \right).
\end{align}
If we choose $K$ as an $N$-dependent sufficiently rapidly growing sequence $(K_N)_{N\ge 1}$, then the Gross transform $U_{K_N}$ has negligible effect on the condensate and the coherent state structure. This is summarized in the next proposition.


\begin{proposition}
\label{proposition: initial bounds for gross transformed product}
Assume $K\ge c $ for some $c > 0 $ and consider the state $\Psi_{N} = U_K^*\ \big( \psi^{\otimes N}  \otimes W(\sqrt N \varphi) \Omega\big)$
with $(\psi, \varphi) \in H^2(\mathbb{R}^3) \times L_1^2(\mathbb{R}^3)$ and $\norm{\psi}_{L^2(\mathbb{R}^3)} = 1$. Then there exists a $C>0$ such that $\sup_{N\in \mathbb N} \big\vert N^{-1}\scp{\Psi_{N}}{H_N^{\rm F} \Psi_{N}} \big\vert \le C$ and
\begin{align}
a(\Psi_{N},\psi) & \le   \frac{C}{K^{3/2}},\quad  b(\Psi_{N},\varphi) \le \frac{C}{K^3}, \quad  c(\Psi_{N}) \le C  \Big( K^{-1} + N^{-1} + \frac{K}{N^2} \Big) \label{INITIAL:BOUND:CN}
\end{align} 
with $a(\Psi_{N},\psi)$, $b(\Psi_{N},\varphi)$ and $c(\Psi_N)$ defined as in Theorem \ref{theorem: main theorem}. 
\end{proposition}

We prove this proposition in Section \ref{sec: initial states}.
As an immediate consequence of Proposition \ref{proposition: initial bounds for gross transformed product} (with $K=cN$) and Theorem \ref{theorem: main theorem} one finds
\begin{equation}
\label{eq: redyced density bound  product state with Gross transform}
\textnormal{Tr}_{L^2(\mathbb{R}^3)} \abs{\gamma^{(1,0)}_{\Psi_{N,t}} - \ket{\psi_t} \bra{\psi_t}}
 \leq   N^{-1/4}   e^{C ( 1+ \abs{t} )^3  }  
\end{equation}
and
\begin{equation}
\frac 1 N \scp{W^*(\sqrt{N} \varphi_t) \Psi_{N,t}}{  {\mathcal N}  \, W^* (\sqrt{N} \varphi_t) \Psi_{N,t}} 
 \leq   N^{-1/2}   e^{C ( 1+ \abs{t} )^3  }
\end{equation}
for initial states of the form $\Psi_N = U_{cN}^* ( \psi^{\otimes N}\otimes W(\sqrt N \varphi)\Omega)$.

Since the quantities $b(\Psi_{N},\varphi)$ and $c(\Psi_{N})$ appearing on the right side of \eqref{MAIN:BOUND:TRACE:NORM:1:0} and \eqref{MAIN:BOUND:TRACE:NORM:0:1} are expectation values of unbounded operators, it is not possible to generalize Theorem \ref{theorem: main theorem} to initial states $\Psi_{N} \notin  \mathcal D( H_N^{\rm F})$ via a simple density argument. Using the Gross transform, however, it is possible to obtain a similar result for initial states in a subset of $\mathcal{D} ( H_N^0 )$. This follows from Theorem \ref{theorem: main theorem} in combination with \eqref{DOMAIN:H:NF} and the fact that $U_{K}$ converges strongly to the identity operator for $K\to \infty$.
The precise statement is as follows. 

\begin{theorem}
\label{theorem: main theorem: 2} Let $K_N \geq   c N^{5/6}$ for some $c >0$. Let $(\psi, \varphi) \in H^2(\mathbb{R}^3) \times L_1^2(\mathbb{R}^3)$ with $\norm{\psi}_{L^2(\mathbb{R}^3)} = 1$, and $\Psi_{N} \in \mathcal{D} (  H_{N}^0 )$ such that $\norm{\Psi_{N}} = 1$ and 
\begin{align}
E_0 = \sup_{N \in \mathbb{N}} \big\vert N^{-1} \scp{\Psi_{N}}{U_{K_N} H_N^{\rm F} U^*_{K_N} \Psi_{N}} \big\vert <\infty.
\end{align}
Let $(\psi_t, \varphi_t)$ be the unique solution of \eqref{eq: Landau Pekar equations} with initial datum $(\psi, \varphi)$ and $\Psi_{N,t}= e^{-iH_N^{\rm F} t}\Psi_{N}$. Then, there exists a constant $C >0$ (depending only on  $c$, $\norm{\varphi}_{L_1^2(\mathbb{R}^3)}$, $\norm{\psi}_{H^2(\mathbb{R}^3)}$ and $E_0$) such that
\begin{align}
\textnormal{Tr}_{L^2(\mathbb{R}^3)} \abs{\gamma^{(1,0)}_{\Psi_{N,t}} - \ket{\psi_t} \bra{\psi_t}} & \le \sqrt{ a( \Psi_{N},\psi)  + b(\Psi_{N},\varphi)  + c (U_{K_N}^* \Psi_{N})  +  N^{-1/2}} e^{C   (1+\abs{t})^3 } ,\label{MAIN:BOUND:TRACE:NORM:1:0:THEOREM:2}  \\[4mm]
& \hspace{-4cm} \scp{W^* (\sqrt{N} \varphi_t) \Psi_{N,t}}{ \sqrt{ \tfrac{\mathcal N}{N} } \, W^* (\sqrt{N} \varphi_t) \Psi_{N,t}} \nonumber \\[1mm]
 &  \le  \sqrt{  a(\Psi_{N},\psi)  + b(\Psi_{N},\varphi)  + c (U_{K_N}^* \Psi_{N})  +  N^{-1/2}}  e^{C   (1+\abs{t})^3 }  \label{MAIN:BOUND:TRACE:NORM:0:1:THEOREM:2}.
\end{align}
In particular, for the Pekar initial state $\Psi_{N} = \psi^{\otimes N} \otimes W(\sqrt{N} \varphi) \Omega$ we have  the bounds
\begin{align}
\label{eq: redyced density bound  product state}
\textnormal{Tr}_{L^2(\mathbb{R}^3)} \abs{\gamma^{(1,0)}_{\Psi_{N,t}} - \ket{\psi_t} \bra{\psi_t}}
& \leq   N^{-1/4}   e^{C ( 1+ \abs{t} )^3  }  , 
\\[1mm]
 \scp{W^* (\sqrt{N} \varphi_t) \Psi_{N,t}}{  \sqrt{ \tfrac{\mathcal N}{N} }  \, W^* (\sqrt{N} \varphi_t) \Psi_{N,t}} 
& \leq   N^{-1/4}   e^{C ( 1+ \abs{t} )^3  }. \label{eq: redyced density bound  product state line 2}
\end{align}
\end{theorem}

 The proof is given in Section \ref{section: proof of the proposition for free domain}.

\begin{remark} The restriction $K_N \geq c N^{5/6}$ was chosen in order to minimize the error terms in \eqref{MAIN:BOUND:TRACE:NORM:1:0:THEOREM:2} and \eqref{MAIN:BOUND:TRACE:NORM:0:1:THEOREM:2}. 
\end{remark}

\begin{remark} Note that in \eqref{MAIN:BOUND:TRACE:NORM:0:1:THEOREM:2} we only control the time evolution of $\sqrt{N^{-1} \mathcal N}$, while in \eqref{MAIN:BOUND:TRACE:NORM:0:1} we estimate the operator $N^{-1}\mathcal N$.
\end{remark}

\section{Preliminaries}\label{sec: Notation and basic estimates}

\subsection{Notation and basic estimates}

We introduce the usual bosonic creation and annihilation operators
\begin{align}
a(f) &= \int d^3k  \, \overline{f(k)} a_k, \quad
a^*(f) = \int d^3k \, f(k) a^*_k , \quad   f \in L^2(\mathbb{R}^3),
\end{align}
as well as the field operators
\begin{align}
\Phi(f) &= a(f) + a^*(f) ,
 \quad
\Pi(f) = \Phi(if) = i \big( - a(f) + a^*(f) \big) .
\end{align}
They satisfy the bounds 
\begin{align}
\label{eq: bound for annihialtion and creation operators}
\norm{a(f) \Psi_{N}} & \leq \norm{f}_{L^2(\mathbb R^3)} \norm{\mathcal{N}^{1/2} \Psi_N}, \quad 
\norm{a^*(f) \Psi_{N}} \leq \norm{f}_{L^2(\mathbb R^3)} \norm{\left(\mathcal{N} + 1 \right)^{1/2} \Psi_N},
\end{align}
and
\begin{align}\label{eq: bound for Pi}
\norm{\Phi(f) \Psi_{N}} &\leq 2 \norm{f}_{L^2(\mathbb R^3)} \norm{\left(\mathcal{N} + 1 \right)^{1/2} \Psi_N},  \quad  
\norm{\Pi(f) \Psi_{N}} \leq 2 \norm{f}_{L^2(\mathbb R^3)} \norm{\left(\mathcal{N} + 1 \right)^{1/2} \Psi_N}
\end{align}
for any $\Psi_N\in \mathcal{H}^{(N)}$. 
For $K >0$, we define the classical fields 
\begin{align}
\Phi_{K}(x,t) &=
 \int_{\abs{k}  \, \leq K} d^3k \,    \abs{k}^{-1} 
\left(  e^{ikx} \varphi_t(k)  + e^{-ikx}  \overline{\varphi_t(k)}  \right)  ,
\nonumber \\[1mm]
 \Phi_{\geq K}(x,t) &=
 \int_{\abs{k}  \, \ge K} d^3k \,     \abs{k}^{-1}
\left(  e^{ikx} \varphi_t(k)  + e^{-ikx}  \overline{\varphi_t(k)}  \right).
\end{align}
Moreover, it is useful to define the functions
\begin{align}
\label{eq: definition G-x-j}
G_{x}(k) &= e^{-i k x} \abs{k}^{-1} , 
\quad   \quad
G_{K, x}(k)  = e^{-i k x} \abs{k}^{-1}  \mathds{1}_{\abs{k} \leq K}(k),
\end{align}
and $B_{K,x}(k) =\frac{-1}{\abs{k} (1 + k^2)} e^{-ikx }   \mathds{1}_{\abs{k} \geq K}(k)$ as in \eqref{eq: definition B-K-N}. The bounds 
\begin{align}
\label{eq: bounds for G}
  \norm{G_{K,x}}_{L^2(\mathbb{R}^3)}^2 &= 4 \pi K,\quad 
  \norm{B_{K,x}}_{L^2(\mathbb{R}^3)}^2 \leq 4 \pi K^{-3}
, \quad 
 \norm{ \abs{ \cdot} B_{K,x}}_{L^2(\mathbb{R}^3)}^2 \leq 4 \pi K^{-1}
\end{align}
are straightforward to verify and will be frequently used in the rest of the article. 
We also have
\begin{align}\label{eq: bound for Phi K}
 \abs{\Phi_K(x,t)} \leq \sqrt{32 \pi} \norm{\varphi_t}_{L_1^2(\mathbb{R}^3) },\quad    \norm{ \Phi (G_{K,x_j})  \Psi_N} &\leq \sqrt{16 \pi K}  \norm{\left(\mathcal{N} + 1 \right)^{1/2} \Psi_N}
\end{align}
for $j\in \{1,...,N\}$.\\

\noindent \textbf{Notation:}
The functions $k\mapsto k G_{K,x}(k)$ and $k\mapsto k B_{K,x}(k)$ will frequently be denoted by $k G_{K,x}$ and $kB_{K,x}$, respectively. Depending on the context $\norm{\cdot}$ and $\scp{\cdot}{\cdot}$ will refer to the norm and scalar product either of $\mathcal{H}^{(N)}$ or $L^2(\mathbb{R}^3)$. If the spaces $L_m^2(\mathbb{R}^3)$ and  $H^m(\mathbb{R}^3)$  (with $m \in \mathbb{N}$) appear  as subscripts we will abbreviate them by $L_m^2$ and $H^m$.

\subsection{Weyl operators and Gross transform }

The Weyl operator $W(f)$ defined in  \eqref{eq: Weyl operator} is unitary, i.e., $W^*(f) = W^{-1}(f)$, and satisfies the relations
\begin{align}
\label{eq: Weyl operators product}
W^{-1}(f)= W(-f), \quad W(f) W(g) = W(g) W(f) e^{-2 i \Im \langle f , g\rangle}  = W(f+g) e^{-i \Im \langle f , g \rangle }
\end{align}
as well as the shift property
\begin{align}
\label{eq: Weyl operators shift property}
W^*(f) a_k W(f) = a_k + f(k).
\end{align}
This immediately implies that the Gross transform, as defined in 
\eqref{eq: definition Gross transform}, is unitary. Moreover, it has the properties
\begin{align}\label{eq: product structure Gross transform}
U_K &= W \Big( - N^{-1/2} \sum_{j=1}^N B_{K,x_j} \Big)
=  \prod_{j=1}^N   W \Big( - N^{-1/2}  B_{K,x_j} \Big)
\end{align}
(which holds since $\Im \scp{B_{K,x}}{B_{K,y}}=0$ for all $x,y\in \mathbb R^3$) and 
\begin{align}
\label{eq: Gross transform shift property}
U_K  a_k U^*_K 
&= a_k + N^{-1/2} \sum_{j=1}^N B_{K, x_j}(k) .
\end{align}

\subsection{The Fr\"ohlich Hamiltonian}
\label{section: construction of the Froehlich Hamiltonian}


In \cite{griesemerwuensch}, Griesemer and W\"unsch give an explicit representation of $H_N^{\rm F}$ with the aid of the Gross transform when $N=1$. Below, we state the analogous representation for $N>1$, which will be useful for the proof of our main theorem. Considering $N>1$ does not impose additional difficulties compared to \cite{griesemerwuensch}. 

\begin{definition}
\label{definition: Gross transformed Froehlich Hamiltonian} With $B_{K,x}$ and $G_{K,x}$  defined  in \eqref{eq: definition B-K-N} and \eqref{eq: definition G-x-j}, respectively, we  set
\begin{align}
\label{eq: Definition of V-K-N}
A_{K,x} &=  - 2 i N^{-1/2}  \Big( \nabla_x \cdot
 a \big( k B_{K,x} \big)
 +  a^* \big( k B_{K,x} \big)  \cdot \nabla_x \Big)
 +  N^{-1}  \Phi\big( k B_{K,x} \big)^2,
 \\[2mm]
 V_{K}(x-y)  & =  N^{-1} \Big(  \scp{B_{K,x}}{B_{K,y}}
+ 2 \Re \scp{G_{x}}{B_{K,y}}   \Big) ,
\\[0.5mm]
H_{N,K}^{\rm F} &= \sum_{j=1}^N  \left[ - \Delta_j  +   N^{-1/2} \Phi(G_{K,x_j}) \right] + \mathcal{N},
\end{align} 
and define the Gross transformed Fr\"ohlich Hamiltonian as
\begin{align}\label{GROSS:TRANSFORMED:HAM}
H^{\rm G}_{N,K}
&=  H_{N,K}^{\rm F} + \sum_{j=1}^N A_{K,x_j} + \sum_{j,l=1}^N  V_K(x_j-x_l).
\end{align}
\end{definition}
Note that  \eqref{eq: bounds for G}  immediately implies the bound
\begin{align}\label{eq: bound for V_K}
\vert  V_K(x_j-x_l) \vert \le C K^{-1}N^{-1} 
\end{align}
for suitable $C>0$.
The following result, which is the generalization of \cite[Theorem 3.7]{griesemerwuensch} to $N\ge 2$, justifies denoting $H_{N,K}^{\rm G}$ as Gross transformed Fr\"ohlich Hamiltonian. 

\begin{lemma}\label{lemma representation of H_F}
The operator $H^{\rm G}_{N,K}$ is self-adjoint on $\mathcal{D} (H_N^0 )$ for all $K >0$. Moreover, there exists a $\widetilde{K}\ge 0$ such that for all $K \geq \widetilde{K}$ and $N \in \mathbb{N}$, the self-adjoint operator $H_N^F$  associated to the quadratic form defined by \eqref{eq: Froehlich Hamiltonian quadratic form} has the representation
\begin{align} \label{REPRESENTATION:H:F}
H_N^{\rm F} = U^*_{K} H_{N,K}^{\rm G} U_{K},
\quad \mathcal{D}(H_N^{\rm F}) = U^{*}_{K} \mathcal D(H_N^0) .
\end{align}
\end{lemma}

We shall comment on the proof of this lemma in Appendix \ref{appendix: interaction bounds}.

For use below, we also note that there is $\widetilde K,C>0$, such that for all $K\ge \widetilde K$ and $N\ge 1$,
\begin{align}
\frac{1}{2} H_{N}^0  - C N  & \le H_{N}^{\rm F} \le \frac{3}{2} H_{N}^0 + C N \label{eq: bound H^F vs H^0},\\[1mm]
\frac{1}{2} H_{N}^0  - C N  & \le H_{N,K}^{\rm G} \le \frac{3}{2} H_{N}^0 + C N \label{eq: bound H^G vs H^0}
\end{align} 
hold as inequalities on the Hilbert space $L^2(\mathbb R^{3N}) \otimes \mathcal F_s$ without symmetry constraints on the particles. This will be useful later in order to estimate expectation values w.r.t. wave functions that are not permutation symmetric in all particle coordinates, as e.g.\ in \eqref{eq: variance lemma first bound}. The derivation of \eqref{eq: bound H^F vs H^0} and \eqref{eq: bound H^G vs H^0} is postponed to Appendix \ref{appendix: interaction bounds}.

\section{Proof of the Main Theorem}
\label{section: Proof of the Main Theorem}

We first state three preliminary lemmas from which the proof of Theorem \ref{theorem: main theorem} then follows easily. The proofs of the lemmas are postponed to later sections.

If we take the limit $K\to \infty$, the Gross transform has only negligible effect on the one-particle reduced density and the coherent structure of the phonon field. This is quantified in the following lemma, whose proof is given in Sec.~\ref{SEC:PROOF:LEMMA:CLOSENESS:GAMMA}. 

\begin{lemma}
\label{LEMMA:CLOSENESS:GAMMA} Assume $K \ge \widetilde K>0$ such that Lemma \ref{lemma representation of H_F} holds. Let $\varphi \in L^2(\mathbb{R}^3)$, $\Psi_N \in \mathcal  D((H_N^{\rm F})^{1/2})$ with $\norm{\Psi_N}=1$, and the Gross transform $U_{K}$ defined as in \eqref{eq: definition Gross transform}. Then,
\begin{align}
\textnormal{Tr}_{L^2(\mathbb R^3)} \Big\vert \gamma^{(1,0)}_{\Psi_N} - \gamma^{(1,0)}_{U_{K} \Psi_N}  \Big\vert  & \le \frac{C}{K^{3/2}}  \norm{ \Big( \frac{H_N^{\rm F} + C N}{N}\Big)^{1/2} \Psi_N } \label{BOUND:CLOSENESS:GAMMA:1:0}
\end{align}
and 
\begin{align}
& N^{-1} \abs{  \scp{ W^*(\sqrt N \varphi) \Psi_{N} }{ \big( \mathcal N - U_K^* \mathcal N U_K \big)  W^* (\sqrt N \varphi)  \Psi_N } }\nonumber  \\
& \hspace{6cm} \le  \frac{C (1+ \norm{\varphi}_2)}{K^{3/2}}  \norm{ \Big( \frac{H_N^{\rm F} + C N}{N}\Big)^{1/2} \Psi_N }\label{BOUND:CLOSENESS:GAMMA:0:1}
\end{align}
for some $C>0$.
\end{lemma}

Next, we define a functional to compare $U_{K}\Psi_{N,t}$ with the Pekar state $\psi_t^{\otimes N}\otimes W(\sqrt N \varphi_t)\Omega$. To this end, we introduce for $j \in \{1,...,N\} $ the projections $p^\psi_j : L^2(\mathbb{R}^{3N}) \rightarrow L^2(\mathbb{R}^{3N})$ and $q_j^{\psi} : L^2(\mathbb{R}^{3N}) \rightarrow L^2(\mathbb{R}^{3N})$, given by
\begin{align}
p_j^{\psi} f_N(x_1, \ldots , x_N) 
&= \psi (x_j) \int d^3 x_j' \, \overline{ \psi(x_j')} f_N(x_1,\ldots x_{j-1},x_j',x_{j+1},\ldots , x_N)
\end{align}
for $f_N \in L^2(\mathbb{R}^{3N})$, and $q_j^{\psi} = 1 - p_j^{\psi}$. (More compactly, in bracket notation: $p^\psi_j  = \vert \psi   \rangle \langle \psi  \vert_j$).

\begin{definition}
Let $K>0$ and $\left(\psi , \varphi \right) \in L^2(\mathbb{R}^3) \times  L^2(\mathbb{R}^3)$ with $\norm{\psi}=1$ and $\Psi_{N}\in \mathcal{D} \left( H_N^{\rm F} \right)$, $\norm{\Psi_{N}} = 1$. We define $\beta^a_K : \mathcal{D} ( H_N^{\rm F} ) \times L^2(\mathbb{R}^3) \rightarrow \mathbb{R}_0^+$,
$\beta^b_K: \mathcal{D} ( H_N^{\rm F} ) \times L^2(\mathbb{R}^3) \rightarrow \mathbb{R}_0^+$ and 
$\beta^c: \mathcal{D} ( H_N^{\rm F} ) \rightarrow \mathbb{R}_0^+$
by 
\begin{align}
\beta^a_K \left( \Psi_{N }, \psi \right)
&= \scp{\Psi_{N }}{U^*_K \left( q_1^{\psi} \otimes \id_{\mathcal{F}_s} \right) U_K \Psi_{N }} ,
\\[1mm]
\beta^b_{K} \left( \Psi_{N }, \varphi\right)
&=
N^{-1} \scp{ W^*(\sqrt N \varphi) U_K \Psi_{N} }{\mathcal N W^*(\sqrt N \varphi) U_K  \Psi_N }
\\[1mm]
\beta^c(\Psi_{N }) &= \norm{N^{-1} \left(  H_N^{\rm F} - \scp{\Psi_{N }}{ H_N^{\rm F} \Psi_{N }} \right) \Psi_{N }}^2.
\end{align}
Moreover, we define $\beta_{K}  : \mathcal{D} \left( H_N^{\rm F} \right) \times L^2(\mathbb{R}^3) \times L^2(\mathbb{R}^3) \rightarrow \mathbb{R}_0^+$ by
\begin{align}
\beta_K(\Psi_N,\psi,\varphi) & = \beta^a_K(\Psi_N,\psi)   + \beta^b_K(\Psi_N,\varphi)   + \beta^c(\Psi_N) .
\end{align}
For solutions $\Psi_{N,t}$ and $(\psi_t,\varphi_t)$ of the Schr\"odinger equation \eqref{eq: Schroedinger equation} and the Landau--Pekar equations \eqref{eq: Landau Pekar equations}, respectively, we use the shorthand notations
\begin{align*}
\beta_K(t) & = \beta_K(\Psi_{N,t},\psi_t,\varphi_t),\ \
\beta^a_K(t) = \beta^a_K (\Psi_{N,t},\psi_t), \ \  
\beta^b_K(t) = \beta^b_K (\Psi_{N,t},\varphi_t), \ \ 
\beta^c(t) = \beta^c (\Psi_{N,t}) .
\end{align*} 
\end{definition}

\begin{remark}
Note that
\begin{align}
\beta^b_K \left( \Psi_{N }, \varphi\right) = \int d^3k \, \norm{\left( N^{-1/2} a_k - \varphi (k) \right) U_K \Psi_{N }}^2.
\end{align}
\end{remark}

\begin{remark}
$\beta_K(t)$ being small compared to one ensures that
\begin{itemize}
\item the $N$-particle component of $U_K \Psi_{N,t}$ is approximately given by the product $\psi_t^{\otimes N}$ -- more precisely, $\beta^a_K(t)$ measures the relative number of particles not in $\psi_t$,
\item the phonon component of $U_K \Psi_{N,t}$ is close to the coherent state $W(\sqrt N \varphi_t)\Omega$ -- more, precisely, $\beta^b_K(t)$ measures the relative number of excitations w.r.t.\ to the coherent state $W(\sqrt N \varphi_t)\Omega$,
\item the variance of $N^{-1} H_N^{\rm F}$ w.r.t.\ $\Psi_{N,t}$ is small compared to one -- this will be used to control the singular ultraviolet behavior of the phonon field (for a detailed explanation of this point, see the beginning of Section \ref{section: variance of the energy bound}). Also note that $\beta^c (\Psi_{N,t}) = \beta^c(\Psi_{N})$ is a conserved quantity, and thus requiring  $\beta^c$ to be small only poses  a restriction on the initial state. Since $\beta^c(\Psi_{N}) = c(\Psi_N)$, Proposition \ref{proposition: initial bounds for gross transformed product} shows that $\beta^c$ is small for initial states of the form $\Psi_N = U_K^* \psi^{\otimes N}\otimes W(\sqrt N \varphi)\Omega$ with $K=K_N$ large enough.
\end{itemize}
The functional $\beta_K(t)$ can consequently be used to monitor whether the condensate of the particles and the coherent state of the phonons is stable during the time evolution. Its definition is motivated by a previous work on the derivation of the Maxwell-Schr\"odinger equations \cite{leopoldpickl}. In addition it is necessary to include the Gross transform in the definition of $\beta^a_K(t)$ and $\beta^b_K(t)$. This induces correlations between the electron and the phonons and effectively regularizes the interaction. In this sense,  the Gross transform has a similar role as the Bogoliubov transformation in the derivation of the time-dependent Gross-Pitaevskii equation (see for instance \cite{pickl2, jeblickleopoldpickl, benedikterschlein, brenneckeschlein}).
\end{remark}

The trace norm of the difference $\gamma^{(1,0)}_{U_K \Psi_{N,t}} - \vert \psi_t \rangle \langle \psi_t \vert$ and the quantity $\beta^a_K(t)$ are related by
\begin{align}
\label{eq: relation between beta and reduced density matrices 1}
\beta^a_K (t) &\leq  \textnormal{Tr}_{L^2(\mathbb{R}^3)} \abs{\gamma_{U_K \Psi_{N,t}}^{(1,0)} - \ket{\psi_t} \bra{\psi_t}} \leq  4 \sqrt{\beta^a_K(t)} ,
\end{align}
which is the content of the following lemma when $\Psi_N = U_K \Psi_{N,t}$.
\begin{lemma}
\label{lemma: relation between beta and reduced density matrices}
Let $\psi \in L^2(\mathbb{R}^3)$ with $\norm{\psi}=1$ and $\Psi_{N} \in \mathcal{H}^{(N)}$ with $\norm{\Psi_N}=1$. Then,
\begin{align}
\label{eq: relation between beta and reduced density matrices 1}
 \scp{\Psi_{N }}{ \big( q_1^{\psi} \otimes \id_{\mathcal{F}_s} \big)   \Psi_{N }}  &\leq  \textnormal{Tr}_{L^2(\mathbb{R}^3)} \abs{\gamma_{\Psi_{N}}^{(1,0)} - \ket{\psi} \bra{\psi}} \leq  4 \sqrt{\scp{\Psi_{N }}{ \big( q_1^{\psi} \otimes \id_{\mathcal{F}_s} \big)   \Psi_{N }} }  .
\end{align}
\end{lemma}

\begin{proof}
The lemma is a consequence of the identity
\begin{align}\label{eq: dual identity}
 \textnormal{Tr}_{L^2(\mathbb{R}^3)} \abs{\gamma_{\Psi_{N}}^{(1,0)} - \ket{\psi} \bra{\psi}}  = \sup_{\norm{A}_{op}=1} \abs{ \scp{\Psi_N}{A_{1} \Psi_N } - \scp{\psi}{A \psi} }
\end{align}
where the supremum is taken over all bounded operators $A:L^2(\mathbb R^3) \to L^2(\mathbb R^3)$ and 
\begin{align}
A_1 = \big( A \otimes \underbrace{\id_{L^2(\mathbb R^3)} \otimes ... \otimes \id_{L^2(\mathbb R^3)}}_{ N -1\ \text{times}} \big) \otimes \id_{\mathcal{F}_s}
\end{align}
acts non-trivially only on the variable $x_1$. (Note that \eqref{eq: dual identity} holds because the space of bounded operators is the dual of the space of trace-class operators). The first bound then follows from\footnote{From now on we omit the product with the identity and write $q_i^\psi$ and $p_i^\psi$ instead of $q_i^\psi \otimes \id_{\mathcal F}$ and $p_i^\psi \otimes \id_{\mathcal F}$.}
\begin{align}
 \scp{\Psi_{N }}{ \big( q_1^{\psi} \otimes \id_{\mathcal{F}_s} \big)   \Psi_{N }} = \scp{ \Psi_N }{ q_1^\psi  \Psi_N} = \abs{ \scp{ \Psi_N}{ p_1^\psi   \Psi_N} - \scp{\psi}{p^\psi \psi}},
\end{align}
while for the second bound, one inserts $1=p_1^\psi+q_1^\psi$ on the left and right of $A_{1}$ and uses 
\begin{align}
p_1^\psi  A_{1} p_1^\psi - \scp{\psi}{A  \psi} = q_1^\psi \scp{\psi}{A \psi}
\end{align}
together with the Cauchy--Schwarz inequality for the remaining terms.
\end{proof}

The main ingredient of the proof of Theorem \ref{theorem: main theorem} is the following estimate for $\beta_K(t)$.

\begin{lemma}
\label{lemma: Gronwall bound for beta}
 Assume $K \ge \widetilde K>0$ such that Lemma \ref{lemma representation of H_F} holds. Let $\Psi_{N,t} = e^{-iH_N^{\rm F} t }\Psi_{N} $ with $\Psi_{N} \in \mathcal{D} ( H_N^{\rm F} )$ such that $\norm{\Psi_{N}} = 1$ and $E_0 = \sup_{N \in \mathbb N} \abs{ N^{-1} \scp{\Psi_{N}}{ H_N^{\rm F} \Psi_{N}} }  <\infty $. Let further $(\psi_t, \varphi_t)$ be a solution of \eqref{eq: Landau Pekar equations} with $(\psi, \varphi)  \in H^2(\mathbb{R}^3) \times L^2_{1}({\mathbb{R}^3})$ such  that $\norm{\psi}=  1$. Then, there exists a constant $C > 0 $ only depending on $\norm{\varphi}_{L_1^2 }$, $\norm{\psi}_{H^2}$, and $E_0$, such that
\begin{align}
\abs{ \frac{d}{dt} \beta_K(t)}
&\leq C \left( 1 + t^2 \right) \left( \beta_K(t )  + K N^{-1} + K^{-1}  \right) . \label{BOUND:BETA:TIME:DERIVATIVE}
\end{align}
\end{lemma}

The proof is given in Section \ref{section: Estimates on the time derivative}.
Putting the above statements together, we obtain the proof of Theorem \ref{theorem: main theorem}.

\begin{proof}[Proof of Theorem \ref{theorem: main theorem}] We first apply Gr\"onwall's argument to \eqref{BOUND:BETA:TIME:DERIVATIVE} in order to obtain
\begin{align}\label{BOUND:BETA}
\beta_K(t ) &\leq e^{C  ( 1+ \abs{t} )^3} \left(  \beta_K(0 ) + K N^{-1} + K^{-1} \right). 
\end{align}
Next, set $K = K_N = \widetilde KN^{1/2}$ with $\widetilde K>0$ as in Lemma \ref{lemma representation of H_F}, and compute
\begin{align}
\text{Tr}_{L^2(\mathbb R^3)} \Big\vert \gamma^{(1,0)}_{\Psi_{N,t}} - \vert \psi_t \rangle \langle \psi_t \vert  \Big\vert  & \le  \text{Tr}_{L^2(\mathbb R^3)} \Big\vert \gamma^{(1,0)}_{U_{K_N}\Psi_{N,t}} - \vert \psi_t \rangle \langle \psi_t \vert  \Big\vert   +  CN^{-3/4}
\nonumber \\[2.5mm]
& \le 4 \sqrt{\beta^a_{K_N}(t)} + C N^{-3/4} 
\nonumber  \\[2.5mm]
& \le 4 \sqrt{ \beta_{K_N}(t)} + C N^{-3/4}
\nonumber  \\[2.5mm]
& \le  \sqrt{\beta_{K_N}(0) +  N^{-1/2}}  e^{C(1+ \abs{t})^3},
\end{align}
where we used inequality \eqref{BOUND:CLOSENESS:GAMMA:1:0} in the first step, Lemma \eqref{lemma: relation between beta and reduced density matrices} in the second and \eqref{BOUND:BETA} in the last one. The estimate \eqref{MAIN:BOUND:TRACE:NORM:1:0} then follows from $\beta^c(0) = c(\Psi_{N})$ and
\begin{align}
\beta^a_{K_N}(0) + \beta^b_{K_N}(0) \le a(\Psi_{N},\psi) + b(\Psi_{N},\varphi) + C N^{-3/4},
\end{align}
which in turn holds because of \eqref{eq: relation between beta and reduced density matrices 1} and Lemma \ref{LEMMA:CLOSENESS:GAMMA}.\medskip

 Using \eqref{BOUND:CLOSENESS:GAMMA:0:1}, we can similarly estimate
\begin{align}
& \hspace{-1cm} N^{-1} \scp{ W^*(\sqrt N \varphi_t)  \Psi_{N,t} }{\mathcal N W^*(\sqrt N \varphi_t) \Psi_{N,t} }
\nonumber  \\[2mm] 
& \le N^{-1} \scp{W^*(\sqrt N \varphi_t) U_{K_N} \Psi_{N,t} }{\mathcal N W^*(\sqrt N \varphi_t) U_{K_N} \Psi_{N,t} }  + C (1+\norm{\varphi_t}) N^{-3/4} 
\nonumber \\[2mm]
& =  \beta^b_{K_N}(t) + C  (1+\norm{\varphi_t}) N^{-3/4}
\nonumber \\[2mm]
& \le \big( \beta_{K_N}(0) + N^{-1/2}\big) e^{C( 1 + \abs{t})^3}
\nonumber \\[2mm]
& \le \big( a(\Psi_{N},\psi) + b(\Psi_{N},\varphi) + c(\Psi_{N}) + N^{-1/2} \big) e^{C( 1+ \abs{t})^3}.
\end{align}
This completes the proof of the theorem.
\end{proof}

 The proofs of Proposition \ref{proposition: initial bounds for gross transformed product} and Theorem \ref{theorem: main theorem: 2} are postponed to Sections \ref{sec: initial states} and \ref{section: proof of the proposition for free domain}, respectively.
 
\section{Bound on $\|  \nabla_2 q_1^{\psi}  U_{K} \Psi_{N}\|$}
\label{section: variance of the energy bound}

In this section, we state and prove a bound that is a crucial ingredient in the proof of Lemma \ref{lemma: Gronwall bound for beta}.

\begin{lemma}
\label{lemma: variance of the energy bound}
Assume $K\ge \widetilde K>0$ such that Lemma \ref{lemma representation of H_F} holds. Let $(\psi, \varphi) \in H^2(\mathbb{R}^3) \times L_1^2(\mathbb{R}^3)$ and $ \Psi_{N} \in \mathcal{D} ( H_N^{\rm F} ) $ with $\norm{\Psi_N}=1$, and set $E_N^{\rm F}(\Psi_N) = N^{-1}    \scp{\Psi_N}{H_N^{\rm F}\Psi_N}$. Then
\begin{align}
\label{eq: variance of the energy bound}
\norm{\nabla_2 q_1^{\psi}  U_{K} \Psi_{N}}^2 
&\leq g(\Psi_N,\psi,\varphi) \left(  \beta_K(\Psi_N,\psi,\varphi)  + N^{-1}  K^{-1} + N^{-2} K \right),
\end{align}
where $g(\Psi_N,\psi,\varphi) = C (\norm{\psi}_{H^2}^2 + \norm{\varphi}_{L_1^2}^2  + \vert  E_N^{\rm F}(\Psi_N) \vert )$ for some $C>0$. 
\end{lemma}

Before we give its proof, we explain the importance of the above estimate.
The main technical difficulty for controlling the time-derivative of $\beta_K(t)$ arises from the singular ultraviolet behavior of the phonon field. In particular, if we want to estimate $\frac{d}{dt}\beta^b_K(t)$, we have to bound the following term (cf.\ Section \ref{subsection time derivative of beta-b})
\begin{align}
\eqref{eq: time derivative beta-b  4}
&= - 2 \Im \scp{U_{K} \Psi_{N,t}}{\int_{\abs{k} \leq K} 
d^3k \, \abs{k}^{-1}  e^{ikx_1} q_1^{\psi_t} \left( N^{-1/2} a_k - \varphi_t(k) \right) U_{K} \Psi_{N,t}} 
\end{align}
by an $N$-independent constant times the functional $\beta_K(t)$. A naive estimate using the Cauchy--Schwarz inequality would give the bound 
\begin{align}\label{eq: naive estimate}
\abs{\eqref{eq: time derivative beta-b  4}}
&\leq C K^{1/2} \sqrt{\beta^a_{K}(t) \beta^b_{K}(t) },
\end{align}
which is not sufficient for $K \gg 1$. The reason for the bad behavior for large $K$ clearly comes from the careless estimate of the form factor $\vert k \vert^{-1}$. 

The most obvious strategy for a better estimate is to apply the well-known commutator method of Lieb and Yamazaki \cite{liebyamazaki}, which utilizes the particle momentum in order to obtain a better ultraviolet behavior of the phonon field. More precisely, one writes the exponential $e^{ik x_1}$ in terms of a commutator with the gradient $i\nabla_{x_1}$,
\begin{align}
e^{ik x_1} = \big( 1 + \abs{k}^{2} \big)^{-1} \left( e^{ik x_1} - k \cdot \big [ i \nabla_{x_1}, e^{ik x_1} \big] \right),
\end{align}
which suggests a better decay for large $\vert k \vert$ provided that one has some control of the regularity of the particle with coordinate $x_1$. Using this identity together with $p_1^{\psi_t} + q_1^{\psi_t} =1$,  we find by a straightforward computation that \eqref{eq: time derivative beta-b  4} can be written as
\begin{align}
- 2 \Im \int_{\abs{k} \leq K} 
d^3k \, \abs{k}^{-1} \big( 1 &+ \abs{k}^{2} \big)^{-1}  \cdot \bigg( 
\scp{e^{-ikx_1} \left( N^{-1/2} a_k^* - \overline{\varphi_t(k)} \right) U_{K} \Psi_{N,t}}{
   q_1^{\psi_t} U_{K} \Psi_{N,t}} 
\nonumber \\
&\; -
\scp{e^{-ikx_1} \left( N^{-1/2} a_k^* - \overline{\varphi_t(k)} \right) k \cdot i \nabla_1 p_1^{\psi_t} U_{K} \Psi_{N,t}}{ q_1^{\psi_t} U_{K} \Psi_{N,t}} 
\nonumber \\
&\; -
\scp{e^{-ikx_1} k \cdot i \nabla_1 q_1^{\psi_t} U_{K} \Psi_{N,t}}{ \left( N^{-1/2} a_k - \varphi_t(k) \right) q_1^{\psi_t} U_{K} \Psi_{N,t}} 
\nonumber \\
&\; + 
\scp{e^{-ikx_1} \left( N^{-1/2} a_k^* - \overline{\varphi_t(k)} \right) U_{K} \Psi_{N,t}}{ k \cdot i \nabla_1 q_1^{\psi_t} U_{K} \Psi_{N,t}} 
\bigg).
\end{align}
With the aid of the Cauchy--Schwarz inequality and the canonical commutation relations this implies the bound
\begin{align}
\abs{\eqref{eq: time derivative beta-b  4}}
&\leq C\left( \norm{\psi_t}_{H^1} \sqrt{\beta^a_K(t)} + \norm{\nabla_1 q_1^{\psi_t} U_{K} \Psi_{N,t}} \right) \sqrt{\beta^b_K(t) + N^{-1}}.
\end{align}
Contrary to \eqref{eq: naive estimate}, there is no more divergence for large $K$. However, the above inequality contains the new term $\| \nabla_1  q_1^{\psi_t} U_{K} \Psi_{N,t} \| $.
Thus if we want to apply Gr\"onwall's inequality we would have to show that this term is  small compared to one or bounded by a constant times $\sqrt{\beta_{K}(t)}$.\footnote{We note that the quantity $\| \nabla_1  q_1^{\psi_t} U_{K} \Psi_{N,t} \|^2$ can be related to the Sobolev trace norm difference between the one-particle reduced density matrix and the condensate wave function (see \cite[Proof of Theorem 2.8]{mitrouskaspetratpickl} 
and \cite[Lemma 7.1]{leopoldpickl2}).}
It is not clear how to derive such a bound, however,  and hence, we are forced to estimate $\abs{\eqref{eq: time derivative beta-b  4}}$ in a different way. 

A possible solution to this problem is to use a combination of the estimates from \cite[Chapter VIII.1]{leopoldpickl} with an operator bound that is motivated by \cite[Lemma 10]{frankschlein} (see Section \ref{section: Estimates on the time derivative} for the detailed argument). In short, we use the symmetry of the wave function and an estimate that is similar in spirit  to the commutator method of Lieb and Yamazaki to obtain 
\begin{align}
\abs{\eqref{eq: time derivative beta-b  4}}
&\leq C \left( \beta^a_{K}(t) + \beta^b_{K}(t) + N^{-1} K + \norm{\nabla_2 q_1^{\psi_t} U_{K} \Psi_{N,t}}^2 \right) .
\end{align}
As shown in Lemma \ref{lemma: variance of the energy bound}, the new quantity $\| \nabla_2 q_1^{\psi_t} U_{K} \Psi_{N,t} \| ^2$ can be bounded by $\beta_{K}(t)$ and errors proportional to $N^{-1}K^{-1}$ and $N^{-2}K$.

\begin{proof}[Proof of Lemma \ref{lemma: variance of the energy bound}] Using the symmetry of $\Psi_{N}$ and $- \Delta_1 \geq 0$, we can bound
\begin{align}\label{eq: variance lemma first bound}
\norm{\nabla_2 q^\psi_1 U_K \Psi_{N}}^2   
& = (N-1)^{-1} \sum_{j=2}^N \scp{q_1^\psi U_K \Psi_{N}}{(-\Delta_j)q_1^\psi U_K \Psi_{N} }\nonumber \\[-1mm]
& \leq  \frac{2}{N} \sum_{j=1}^N \scp{q_1^\psi U_K \Psi_{N}}{(-\Delta_j)q_1^\psi U_K \Psi_{N} }.
\end{align}
With $-\sum_{j=1}^N \Delta_{j} \le H_N^0$ and  \eqref{eq: bound H^G vs H^0}, we thus have 
\begin{align}
\norm{\nabla_2 q_1^\psi U_K \Psi_{N}}^2 \le C  \beta^a_K(\Psi_N,\psi) +  4 N^{-1} \scp{q_1^\psi U_K \Psi_{N}}{  H_{N,K}^{\rm G}  q_1^\psi  U_K  \Psi_{N} } .
\end{align}
By using
$
q_1^\psi H_{N,K }^{\rm G} q_1^\psi = q_1^\psi H_{N,K }^{\rm G} - q_1^\psi H_{N,K}^{\rm{G}} p_1^\psi
$
and  recalling Definition \eqref{definition: Gross transformed Froehlich Hamiltonian}, we get
\begin{subequations}
\begin{align}
\norm{\nabla_2 q_1^\psi U_K \Psi_{N}}^2 
&\leq C  \Big(  \beta^a_K(\Psi_N,\psi) + N^{-1} \abs{ \scp{q_1^\psi U_K \Psi_{N}}{  H_{N,K}^{\rm G} U_K \Psi_{N} }  }   \label{LINE:Q:HG} \\[1.5mm]
&\quad +  N^{-1} \abs{\scp{U_K \Psi_{N}}{q_1^\psi \left( - \Delta_1 \right) p_1^\psi  U_K \Psi_{N}}}\label{LINE:Q:DELTA:P} \\[1.5mm]
&\quad +  N^{-3/2} \abs{\scp{U_K \Psi_{N}}{q_1^\psi  \Phi(G_{K,x_1}) p_1^\psi U_K \Psi_{N}} }  \label{LINE:Q:PHI:P} \\[1.5mm]
&\quad +  N^{-1} \abs{\scp{U_K \Psi_{N}}{q_1^\psi A_{K ,x_1} p_1^\psi U_K \Psi_{N}} }  \label{LINE:Q:V:P}  \\[1.5mm]
&\quad + \abs{\scp{U_K \Psi_{N}}{ q_1^\psi  V_K(x_1-x_2)  p_1^\psi U_K \Psi_{N}} } \Big) .  \label{LINE:Q:C:P} 
\end{align}
\end{subequations}
In the following, we shall bound the various terms on the right hand side. \\ \\
\textbf{Line \eqref{LINE:Q:HG}}. In the second summand in this line, we add and subtract $E_N^{\rm F}(\Psi_N) \beta^a_K(\Psi_N,\varphi)$,
 to obtain 
\begin{align}
&N^{-1}\abs{ \scp{q_1^\psi U_K \Psi_{N}}{  H_{N,K}^{\rm G} U_K \Psi_{N} } } 
\nonumber \\[1mm]
& \quad \quad \le  \abs{ \scp{q_1^\psi U_K \Psi_{N}}{ \left( N^{-1} H_{N,K}^{\rm G}   -  E_N^{\rm F}(\Psi_N) \right)    U_K \Psi_{N} } }  + E_N^{\rm F}(\Psi_N)  \beta^a_K(\Psi_N,\psi).
\end{align}
With the aid of  the Cauchy--Schwarz inequality and \eqref{REPRESENTATION:H:F}, we find
\begin{align}
\abs{\eqref{LINE:Q:HG}} \le C (1+ \vert E_N^{\rm F}(\Psi_N) \vert ) \left( \beta_K^a(\Psi_N,\psi)  + \beta^c(\Psi_N) \right) .
\end{align}
\textbf{Line \eqref{LINE:Q:DELTA:P}.} One readily obtains
\begin{align}
\eqref{LINE:Q:DELTA:P} & \leq  N^{-1} \norm{\left( - \Delta_1 \right) p_1^\psi U_K \Psi_{N} } \norm{q_1^\psi U_K \Psi_{N}} \leq \frac 12 \norm{\psi}_{H^2} \left( \beta^a_K(\Psi_N,\psi) + N^{-2} \right) .
\end{align}
\textbf{Line \eqref{LINE:Q:PHI:P}.} Using \eqref{eq: bound for Phi K}, we find
\begin{align}
\eqref{LINE:Q:PHI:P} & \leq   N^{-3/2} \norm{q_1^\psi U_K \Psi_{N}} \norm{ \Phi \left( G_{K,x_1} \right) p_1^\psi U_K \Psi_{N}}
\nonumber \\[1mm]
&  \leq  C N^{-3/2} K^{1/2} \sqrt{\beta^a_K(\Psi_N,\psi)\,  \scp{U_K \Psi_{N}}{(  \mathcal{N} +1)  U_K \Psi_{N}}} 
\end{align}
and hence, using  \eqref{eq: bound H^G vs H^0},  \eqref{REPRESENTATION:H:F} and $\mathcal N \le H_N^0$,  we have
\begin{align}
\eqref{LINE:Q:PHI:P} & \leq  C  (1 + \vert E_N^{\rm F}(\Psi_N) \vert )  \left( KN^{-2} + \beta^a_K(\Psi_N,\psi) \right) .
\end{align}
\textbf{Line \eqref{LINE:Q:V:P}.} We recall the definition of $A_{K,x}$ in \eqref{eq: Definition of V-K-N} and estimate the term with $a^*(kB_{K,x}) \cdot \nabla_x$ by
\allowdisplaybreaks
\begin{align}
& \abs{ \SCP{U_K \Psi_{N}}{q_1^\psi N^{-1/2} a^*(kB_{K,x_1}) \cdot  \nabla_1  p_1^\psi U_K  \Psi_{N}} } 
\nonumber  \\[2mm]
&\quad  \le  \int d^3k \ \abs{ \SCP{q_1^\psi \Big( N^{-1/2} a_k - \varphi(k) + \varphi(k) \Big) U_K \Psi_{N}}{  k B_{K,x_1}(k)  \cdot \nabla_1 p_1^\psi  U_K \Psi_{N}} }  \nonumber \\[1mm]
&\quad \le   \norm{\nabla_1 p_1^\psi  U_K \Psi_{N}} \int d^3k \left( \abs{ k B_{K,x} (k)}  \norm {\Big( N^{-1/2} a_k - \varphi(k) \Big) U_K \Psi_{N}}  +   \abs{k B_{K,x} (k)  \varphi(k)} \right) \nonumber\\[1mm]
&\quad \le  \norm{\psi}_{H^1 }   \left( \norm{ \abs{\cdot} B_{K,x} } \, \sqrt{\beta_K^b(\Psi_N,\varphi) } +  \norm{ B_{K,x} } \, \norm{\abs{\cdot}\varphi }  \right) \nonumber\\
&\quad \le C \norm{\psi}_{H^1 } \left(  K^{-1/2}  \sqrt{ \beta^b_K(\Psi_N,\varphi)}  + K^{-3/2}   \norm{\varphi}_{L_1^2}   \right).
 \label{eq: nabla a term in AK}
\end{align}
Using $q_1^\psi = 1-p_1^\psi$ and $-\Delta_1\le N^{-1} H_N^0$ as quadratic forms on $L^2_s(\mathbb R^{3N}) \otimes \mathcal{F}_s $, together with \eqref{eq: bound H^G vs H^0}, we find
\begin{align}\label{eq: bound for nabla_1 q_1}
\norm{\nabla_1 q_1^\psi U_K \Psi_{N}}^2 \le 2 \left( \norm{\nabla_1 p_1^\psi U_K \Psi_{N}}^2 + \norm{\nabla_1 U_K \Psi_{N}}^2\right ) \le  C \left( \norm{\psi}_{H^1 }^2 +  \vert E_N^{\rm F} (\Psi_N) \vert + 1\right).
\end{align} 
With this at hand, we can proceed for the term with $\nabla_x \cdot a(kB_{K,x})$ similarly as in \eqref{eq: nabla a term in AK}, with the result that
 \begin{align}
 &   \abs{  \SCP{U_K \Psi_{N} }{q_1^\psi 
\nabla_1 \cdot  N^{-1/2} a(kB_{K,x_1})  p_1^\psi U_K \Psi_{N} } }
\nonumber  \\[1mm]
& \le  \norm{\nabla_1 q_1^\psi U_K \Psi_N} \int d^3k \left(  \abs{kB_{K,x}(k)}    \norm{ \Big( N^{-1/2}a_k - {\varphi(k)} \Big) U_K \Psi_{N} }  + \abs{k B_{K,x}(k) \varphi(k) }  \right)
\nonumber \\[1mm]
& \le  C \sqrt{\norm{\psi}_{H^1 }^2 +  \vert E_N^{\rm F} (\Psi_N)  \vert } \left(  K^{-1/2}  \sqrt{ \beta^b_K(\Psi_N,\varphi)}  + K^{-3/2}   \norm{\varphi}_{L_1^2}   \right).
\end{align}
Next, we estimate the term  in line \eqref{LINE:Q:V:P} with $\Phi(kB_{K,x})^2$,
\begin{align}
\abs{  \SCP{U_K \Psi_{N} }{q_1^\psi N^{-1}   \Phi( kB_{K,x_1} )^2 p_1^\psi U_K \Psi_{N} } }   &  \le N^{-1} \norm{ \Phi( kB_{K,x_1} ) q_1^\psi U_K\Psi_{N}} \,  \norm{ \Phi( kB_{K,x_1} ) p_1^\psi U_K \Psi_{N}}  
\nonumber \\[2mm]
& \le C N^{-1}   \norm{\abs{\cdot} B_{K ,x}}^2 \scp{   U_K \Psi_{N}}{(\mathcal N+1) U_K \Psi_N } \nonumber \\[2.5mm]
& \le C K^{-1} N^{-1} \Big( \scp{U_K \Psi_{N}}{  ( H_{N,K }^{\rm G} +C N )   U_K \Psi_{N} } + 1 \Big)
\nonumber \\[2mm]
& \le C \left( \vert E_N^{\rm F}(\Psi_N) \vert + 1 \right)  K^{-1}  .
\end{align}
By summing  up the terms, we obtain the bound
\begin{align}
&  \abs{\scp{U_K \Psi_N}{ q_1^\psi A_{K,x_1} p_1^\psi U_K \Psi_N} }  
\nonumber \\[1mm]
&\hspace{1cm} \le C \big(\norm{\psi}_{H^1}^2 + \norm{\varphi}_{L_1^2}^2 + \vert E_N^{\rm F}(\Psi_N) \vert \big)    \left( K^{-1} +  \beta_K^b(\Psi_N,\varphi) \right) .\label{eq: bound for pAq}
\end{align}
%
\textbf{Line \eqref{LINE:Q:C:P}.} Using \eqref{eq: bound for V_K},
\begin{align}
\abs{ \scp{U_K\Psi_{N}}{ q_1^\psi V_K(x_1-x_2)  p_1^\psi U_K \Psi_{N}} } & \le \sqrt{ \beta^a_K(\Psi_N, \psi) } \norm{V_K(x_1-x_2)  p_1^\psi U_K \Psi_{N}} \nonumber \\[2mm]
& \le C\big( \beta^a_K(\Psi_N,\psi) +  N^{-2}   K^{-2}\big) .\label{eq: bound for V(x1-x2)}
\end{align}
This completes the proof of the lemma.
\end{proof}

\section{Proof of Lemma \ref{lemma: Gronwall bound for beta}  (time derivative of $\beta_K(t)$)}
\label{section: Estimates on the time derivative}
We first observe that 
\begin{align}
\label{eq: time evolution Gross transformed state}
\frac{d}{dt} U_{K} \Psi_{N,t} 
&= - i U_{K} H^{\rm F}_N \Psi_{N,t}
= - i U_{K} H^{\rm F}_N  U_{K}^* U_{K}   \Psi_{N,t}
= - i H_{N,K}^{\rm G} U_{K} \Psi_{N,t},
\end{align}
from which it follows readily that $\frac{d}{dt}\beta^c(t) = 0$. The time-derivatives of $\beta^a_K(t)$ and $\beta^b_K(t)$ are estimated in the next two sections. Throughout both sections, we use the abbreviation $E_N^{\rm F}(\Psi_{N}) =N^{-1}  \scp{\Psi_N}{H_N^{\rm F}\Psi_N}   $.



\allowdisplaybreaks

 \subsection{Time derivative of $\beta^a_K(t)$}
 
For $q_1^t= q_1^{\psi_t} = 1 - p_1^{\psi_t}$, we have
\begin{align}
\frac{d}{dt} q_1^{t} = - \frac{d}{dt}p_1^{t} = i \big[-\Delta_1 + \Phi(x_1,t) , p_1^{t}\big] = - i \big[-\Delta_1 + \Phi(x_1 ,t) , q_1^{t}\big] .
\end{align}
Using this together with \eqref{eq: time evolution Gross transformed state}, we compute
\begin{align}
\frac{d}{dt}\beta^a_K (t) & = \frac{d}{dt} \SCP{U_K \Psi_{N,t}}{q_1^t U_K\Psi_{N,t}}\nonumber \\[1.5mm]
& = - 2\Im \scp{U_K \Psi_{N,t}}{ \left( H_{N,K}^{\rm G}  + \Delta_1 -  \Phi(x_1,t) \right)  q_1^t  U_K \Psi_{N,t}}\nonumber \\[1.5mm]
& = - 2\Im \scp{U_K \Psi_{N,t}}{ p_1^t\left( H_{N,K}^{\rm G}  + \Delta_1 -  \Phi(x_1,t) \right)  q_1^t  U_K \Psi_{N,t}}  , \label{VI.3}
\end{align}
where we inserted $1= p_1^t+q_1^t$ and used that the term with $q_1^t$ on both sides is real. Recall Definition \ref{definition: Gross transformed Froehlich Hamiltonian}. Using $\Phi(x_1,t) = \Phi_{K}(x_1,t) + \Phi_{\ge K}(x_1,t)$, $p_1^tq_1^t = 0$ and the symmetry of $\Psi_N$, we can rewrite \eqref{VI.3} as
\begin{subequations}
\begin{align}
\frac{d}{dt}\beta^a_K (t) & =  - 2\Im \scp{U_K \Psi_{N,t}}{ p_1^t\left( N^{-1/2} \Phi(G_{K,x_1} ) - \Phi_{K}(x_1,t)  \right)  q_1^t U_K \Psi_{N,t}}   \label{EST:TD:BETA:A:a}\\[2mm]
& \quad + 2\Im  \scp{U_K \Psi_{N,t}}{ p_1^t \Phi_{\ge K}(x_1,t)  q_1^t U_K \Psi_{N,t}} \label{EST:TD:BETA:A:b} \\[3.5mm]
& \quad - 2\Im \scp{U_K \Psi_{N,t}}{ p_1^t  A_{K,x_1} q_1^tU_K \Psi_{N,t}}\label{EST:TD:BETA:A:c}\\[3.5mm]
& \quad - 2\Im \scp{U_K \Psi_{N,t}}{ p_1^t (N-1) V_K(x_1-x_2) q_1^t  U_K \Psi_{N,t}} .\label{EST:TD:BETA:A:d}
\end{align}
\end{subequations}
The various terms will be bounded as follows. \\ \\
\noindent \textbf{Line \eqref{EST:TD:BETA:A:a}.}
We bound 
\begin{align}
\abs{\eqref{EST:TD:BETA:A:a} } & \leq   2  \abs{ \scp{\int_{\vert k \vert \le K } d^3k \, \vert k \vert^{-1} e^{- ikx_1}   \left(  N^{-1/2} a^*_k - \overline{\varphi_t(k)} \right)  p_1^t U_K \Psi_{N,t}}{ q_1^t U_K \Psi_{N,t}}  }
\nonumber \\
& + 2 \abs{ \scp{ \int_{\vert k \vert \le K  }  d^3k \, \abs{k}^{-1}    e^{ikx_1}   \left( N^{-1/2} a_k - \varphi_t(k) \right)  p_1^t U_K\Psi_{N,t}}{ q_1^t U_K \Psi_{N,t}} }
\nonumber \\
& \le  2 \beta^a_K (\Psi_{N,t},\psi_t)
+ \norm{\int_{\vert k \vert \le K } d^3k 
\, \abs{k}^{-1}   e^{-ikx_1} \left( N^{-1/2} a_k^* - \overline{\varphi_t(k)} \right)  p_1^t U_K \Psi_{N,t}}^2 
\nonumber \\
&\quad +  \norm{ \int_{\vert k \vert \le K } d^3k \, \abs{k}^{-1}   e^{i kx_1} \left( N^{-1/2} a_k  -  \varphi_t(k) \right)  p_1^t U_K \Psi_{N,t} }^2.\label{EST:TD:BETA:A:a:SECOND:LINE}
\end{align}
For the second summand, we use
\begin{align}
& \norm{\int_{\vert k \vert \le K }   d^3k \, \abs{k}^{-1}  e^{-ikx_1} \left( N^{-1/2} a_k^* - \overline{\varphi_t(k)} \right)  p_1^t U_K \Psi_{N,t}}^2  \nonumber \\
& \hspace{2.5cm} \le \frac{C K}{N}  +  \norm{ \int_{\vert k \vert \le K } d^3 k \,  \vert k \vert^{-1}   e^{ikx_1} \left( N^{-1/2} a_k  -  \varphi_t(k) \right)  p_1^t U_K  \Psi_{N,t} }^2 ,
\end{align}
which follows directly from the canonical commutation relations. By the shift property \eqref{eq: Weyl operators shift property}, the last summand in \eqref{EST:TD:BETA:A:a:SECOND:LINE} can be written as
\begin{align}
N^{-1} \norm{ a \left( G_{K,x_1} \right)   p_1^t W^*(\sqrt{N} \varphi_t) U_K \Psi_{N,t} }^2 .
\end{align}
In order to estimate this expression, we use \cite[Lemma 10]{frankschlein} which implies the bound
\begin{align}
a^*\left( G_{K,x_1} \right) a \left( G_{K,x_1} \right)
&\leq C_{G}  \, \left( 1 - \Delta_1 \right) \mathcal{N} ,\label{eq: bound for a*a analogous to Fran and Schlein}
\end{align}
with 
\begin{align}\label{eq: integral rearrangement}
C_G = \sup_{p\in \mathbb R^3} \int_{\mathbb R^3} \frac{d^3k}{k^2(1+(p+k)^2)}=\int_{\mathbb R^3} \frac{d^3k}{k^2(1+k)^2} <\infty.
\end{align}
The latter is obtained via a rearrangement inequality. In combination, we thus have 
\begin{align}
\abs{
\eqref{EST:TD:BETA:A:a} }
&\leq C \left(  \beta^a_K (t) + K N^{-1} \right)  +  C_G  \, N^{-1}  \norm{\left( 1 - \Delta_1 \right)^{1/2} p_1^t \mathcal{N}^{1/2}  W^* ( \sqrt{N} \varphi_t ) U_K \Psi_{N,t}}^2
\nonumber \\[1mm]
&\leq C \left(  \beta^a_K(t ) + K N^{-1} \right)  +  C_G  \, \norm{\psi_t}^2_{H^1} N^{-1}   \norm{ \mathcal{N}^{1/2}  W^* ( \sqrt{N} \varphi_t ) U_K \Psi_{N,t}}^2
\nonumber \\[1mm]
&\leq C \norm{\psi_t}^2_{H^1}  \left(  \beta^a_K(t ) +  \beta^b_K(t ) + K N^{-1}  \right).
\end{align}
\noindent \textbf{Line \eqref{EST:TD:BETA:A:b}.} This term  can be estimated as
\begin{align}
\abs{\eqref{EST:TD:BETA:A:b}}  &\leq C
\sup_x \abs{\Phi_{\ge K}(x,t)} \norm{q_1^t U_K \Psi_{N,t}}
\nonumber \\[1.8mm]
&\leq C \sqrt{\beta^a_K (t)} \int_{\abs{k} \geq K} d^3k  \abs{k}^{-1}  \abs{\varphi_t} 
\nonumber \\[-0.8mm]
&\leq C \sqrt{\beta^a_K (t)} \norm{\varphi_t}_{L_1^2}
\left(\int_{\abs{k} \geq K} d^3k  \abs{k}^{-4} \right)^{1/2} \leq C \norm{\varphi_t}_{L_1^2} \sqrt{\beta^a_K (t)} 
K^{-1/2}.
\end{align} 
\noindent \textbf{Line \eqref{EST:TD:BETA:A:c}.} It follows from \eqref{eq: bound for pAq} that 
\begin{align}
\abs{\eqref{EST:TD:BETA:A:c} } \le C \Big( \norm{\psi_t}_{H^1}^2 + \norm{\varphi_t}_{L_1^2}^2 + \vert E_N^{\rm F}(\Psi_{N,t}) \vert \Big)    \left( K^{-1} +  \beta_K^b(t) \right).
\end{align}
\textbf{Line \eqref{EST:TD:BETA:A:d}.} In analogy to  \eqref{eq: bound for V(x1-x2)}  one finds
\begin{align}
\abs{\eqref{EST:TD:BETA:A:d}} \le C   \big( \beta^a_K(t ) + K^{-2}\big) .
\end{align}
In combination, we have thus shown that 
\begin{align}\label{eq: summary time derivative beta a}
\abs{\frac{d}{dt}\beta^a_K (t)} \le C \left( \norm{\psi_t}_{H^1}^2 + \norm{\varphi_t}_{L_1^2}^2 + \vert E^{\rm F}_N(\Psi_{N,t})  \vert \right) \left( \beta_K^a(t) + \beta^b_K(t) + \frac{K}{N} + K^{-1} \right).
\end{align}

\subsection{Time derivative of $\beta^b_K(t)$}
\label{subsection time derivative of beta-b}

From \eqref{eq: time evolution Gross transformed state} we get 
\begin{subequations}
\begin{align}
\label{eq: time derivative of beta-b product rule a}
& \hspace{-0.25cm} \frac{d}{dt} \beta_K^b(\Psi_{N,t},\varphi_t) = \nonumber \\
& =   \int d^3k \, \frac{d}{dt} \scp{\left( N^{-1/2} a_k - \varphi_t(k)\right) U_K \Psi_{N,t}}{\left( N^{-1/2} a_k - \varphi_t(k) \right) U_K \Psi_{N,t}}
\nonumber \\
&= -  2 \Re  \int d^3k \,  \scp{ \left( N^{-1/2} a_k^* - \overline{\varphi_t(k)} \right) \left( N^{-1/2} a_k - \varphi_t(k) \right) U_K \Psi_{N,t}}{ i  H_{N,K}^{\rm G}  U_K \Psi_{N,t}}
\\
\label{eq: time derivative of beta-b product rule b}
&\quad
- 2 \Re \int d^3k \,  \scp{\left( \partial_t \varphi_t(k)  \right) U_K \Psi_{N,t}}{\left( N^{-1/2} a_k - \varphi_t(k) \right) U_K \Psi_{N,t}} ,
\end{align}
\end{subequations}
which is a slightly formal computation, since the use of the product rule of differentiation is not completely obvious here. We clarify the difficulty and justify the above identity in detail in Appendix \ref{section: rigorous derivation of the time derivative}. Next, we write the first line in terms of the commutator as
\begin{align}
\label{eq: use of commutator}
\eqref{eq: time derivative of beta-b product rule a} 
& = - i  \int d^3k \,   \scp{  \big[ H_{N,K}^{\rm G} , \big( N^{-1/2} a_k^* - \overline{\varphi_t(k)} \big) \big(  N^{-1/2} a_k - \varphi_t(k) \big)  \big]   U_K \Psi_{N,t}}{   U_K \Psi_{N,t}}  .
\end{align}
Let us remark that the right hand side is well defined since the commutator of $H_{N,K}^{\rm G}$ and
\begin{align}
L_N = N^{-1} W \big( \sqrt{N} \varphi_t \big) \mathcal{N} W^* \big( \sqrt{N} \varphi_t \big) 
=  \int d^3k \, \big( N^{-1/2} a_k^* - \overline{\varphi_t(k)} \big) \big(  N^{-1/2} a_k - \varphi_t(k) \big)
\end{align}
defines a bounded operator from $\mathcal{D} \big( H_{N,K}^{ \rm G} \big) = \mathcal{D} \big( H_N^{0} \big)$ to $\mathcal{H}^{(N)}$. This fact is a direct consequence of the B.L.T. theorem because $\big[ H_{N,K}^{\rm G}, L_N \big]$ is a bounded operator from $\mathcal{D} \big(  (H_N^0 )^2 \big)$ to $\mathcal{H}^{(N)}$ and the estimate
\begin{align}
\norm{\big[ H_{N,K}^{\rm G}, L_N \big] \Psi_N} \leq C_{N,K} \norm{\big(1 + H_{N}^0 \big) \Psi_N}
\end{align}
holds for all $\Psi_N \in \mathcal{D} \big( (H_N^0)^2 \big)$ and some constant $C_{N,K}$. The latter is straightforward to verify with the aid of
\begin{align}
\label{eq: commutator H a}
\left[ H_{N,K}^{\rm G} , a_k   \right] 
&= 2 N^{-1/2} \sum_{j=1}^N B_{K,x_j}(k)
k \cdot \left( i \nabla_j - N^{-1/2} \Phi \big( k B_{K,x_j} \big)  \right)
\nonumber \\
&\quad  - a_k 
- N^{-1/2} \sum_{j=1}^N \abs{k}^{-1} \id_{\abs{k} \leq K}(k) e^{-ikx_j}
\end{align}
and the basic estimates from Section \ref{sec: Notation and basic estimates}.

Hence, we can proceed by using \eqref{eq: use of commutator} and \eqref{eq: commutator H a} together with the Landau--Pekar equations \eqref{eq: Landau Pekar equations} and the symmetry of the many-body wave function in order to obtain
\small
\begin{align}
\label{eq: derivative of beta-b link}
&\frac{d}{dt}  \beta^b_K (\Psi_{N,t}, \varphi_t)
\nonumber \\
&\quad= - 2 \int_{\abs{k} \leq K} d^3k \, \abs{k}^{-1}  \Im \scp{e^{-ikx_1} U_K \Psi_{N,t}}{\left( N^{-1/2} a_k - \varphi_t(k)  \right) U_K \Psi_{N,t}}
\nonumber \\
&\qquad + 2 \int d^3k \, \abs{k}^{-1}    \Im \scp{ \int d^3y \, e^{-iky} \abs{\psi_t(y)}^2 U_K \Psi_{N,t}}{\left( N^{-1/2} a_k - \varphi_t(k)  \right) U_K \Psi_{N,t}}
\nonumber \\
&\qquad + 4 \int d^3k \,   \Im \scp{ B_{K,x_1}(k) k \cdot
\left( i \nabla_1 - N^{-1/2} \Phi(k B_{K,x_1})  \right) U_K \Psi_{N,t}}{\left( N^{-1/2} a_k - \varphi_t(k)  \right) U_K \Psi_{N,t}}  .
\end{align}
\normalsize Finally, inserting the identity $e^{ikx_1} = p_1^t e^{ikx_1} p_1^t  + q_1^t e^{ikx_1} p_1^t + e^{ikx_1} q_1^t$ leads to
\footnotesize
\begin{subequations}
\begin{align}
& \frac{d}{dt}  \beta^b_K(\Psi_{N,t}, \varphi_t) =
\nonumber \\
\label{eq: time derivative beta-b  1}
&\quad \hspace{-0.4cm}  = - 2 \int_{\abs{k} \leq K} d^3k \, \abs{k}^{-1}  \Im \scp{ U_K \Psi_{N,t}}{ \Big( p_1^t e^{ikx_1} p_1^t  - \int d^3y \, e^{iky} \abs{\psi_t(y)}^2 \Big) \left( N^{-1/2} a_k - \varphi_t(k)  \right) U_K \Psi_{N,t}}
\\
\label{eq: time derivative beta-b  2}
&\qquad \hspace{-0.4cm}  +
2 \int_{\abs{k} \geq K} d^3k \, \abs{k}^{-1} 
\Im \scp{U_K \Psi_{N,t}}{\int d^3y \, e^{iky} \abs{\psi_t(y)}^2 \left( N^{-1/2} a_k - \varphi_t(k) \right) U_K \Psi_{N,t}}
\\
\label{eq: time derivative beta-b  3}
&\qquad \hspace{-0.4cm}  - 2 \int_{\abs{k} \leq K} d^3k \,   \abs{k}^{-1} 
\Im \scp{U_K \Psi_{N,t}}{q_1^t e^{ikx_1} p_1^t \left( N^{-1/2} a_k - \varphi_t(k) \right) U_K \Psi_{N,t}}
\\
\label{eq: time derivative beta-b  4}
&\qquad\hspace{-0.4cm}  - 2 \int_{\abs{k} \leq K} 
d^3k \, \abs{k}^{-1} \Im \scp{U_K \Psi_{N,t}}{e^{ikx_1} q_1^t \left( N^{-1/2} a_k - \varphi_t(k) \right) U_K \Psi_{N,t}}
\\
\label{eq: time derivative beta-b  5}
&\qquad \hspace{-0.4cm}  +
4 \int d^3k \,  \Im
\scp{B_{K,x_1}(k) k \cdot \left( i \nabla_1 - N^{-1/2} \Phi \big( k B_{K,x_1} \big) \right) U_K \Psi_{N,t}}{\left( N^{-1/2} a_k - \varphi_t(k) \right) U_K \Psi_{N,t}}.
\end{align}
\end{subequations}
\normalsize
In the following, we estimate each term on the r.h.s. separately.\\
\\
\noindent \textbf{Line \eqref{eq: time derivative beta-b  1}.} The first term is the most important because it is the one where the particle  density cancels the source term of the Landau--Pekar equations.
We first note that
\begin{align}
p_1^t e^{ikx_1} p_1^t - \int d^3y \, e^{iky} \abs{\psi_t(y)}^2
= \left( p_1^t - 1 \right)  \int d^3y \, e^{iky} \abs{\psi_t(y)}^2
= - q_1^t  \scp{\psi_t}{ e^{ ik \cdot} \psi_t}.
\end{align}
We then use
$
e^{ikx} = \frac{1 - i  (k \cdot \nabla_x)}{1 + k^2} e^{ikx}
$
and  integrate by parts to obtain the bound
\begin{align}
\label{eq: Lieb Yamazaki for expectation value of e-ikx}
\abs{\scp{\psi_t}{e^{ik \cdot} \psi_t}}
& \leq  \abs{\scp{\frac{1 - i ( k \cdot \nabla) }{1 + k^2} \psi_t}{e^{ik \, \cdot \,} \psi_t} }
+ \abs{\scp{\psi_t}{e^{ik \, \cdot \,} \frac{ i  ( k \cdot \nabla)}{1 + k^2} \psi_t} }
\nonumber \\[1mm]
&  \leq 2 (1 + k^2)^{-1} \left( 1 + \abs{k} \norm{\nabla \psi_t} \right) 
\leq 2 \norm{\psi_t}_{H^1} (1 + k^2)^{-1} \left( 1 + \abs{k} \right) .
\end{align} 
Hence,
\begin{align}
\abs{ \eqref{eq: time derivative beta-b  1} }
&\leq 4 \norm{\psi_t}_{H^1} \int d^3k \, \frac{(1+ \abs{k})}{\abs{k} ( 1 + k^2)}
\abs{\scp{q_1^t U_K  \Psi_{N,t}}{\left( N^{-1/2} a_k - \varphi_t(k) \right) U_K \Psi_{N,t}}}
\nonumber \\
&\leq  4 \norm{\psi_t}_{H^1} \sqrt{\beta^a_K(t)} \Big( \int d^3k \, \frac{(1+ \abs{k})^2}{\abs{k}^2 (1+ k^2)^2} \Big)^{1/2}
\Big( \int d^3k \,  \norm{\left( N^{-1/2} a_k - \varphi_t(k) \right) U_K \Psi_{N,t}}^2 \Big)^{1/2}
\nonumber \\[1mm]
&\leq C \norm{\psi_t}_{H^1} \beta_K(t) .
\end{align}

\noindent \textbf{Line \eqref{eq: time derivative beta-b  2}.} We again use \eqref{eq: Lieb Yamazaki for expectation value of e-ikx}
and estimate
\begin{align}
\eqref{eq: time derivative beta-b  2}
&\leq 2 \int_{\abs{k} \geq K} d^3k \, \abs{k}^{-1} 
\abs{\scp{\psi_t}{e^{ik \cdot} \psi_t}} \abs{\scp{U_K \Psi_{N,t}}{\left( N^{-1/2} a_k - \varphi_t(k) \right) U_K \Psi_{N,t}}}
\nonumber \\
&\leq C \norm{\psi_t}_{H^1} \int_{\abs{k} \geq K} d^3k \,   \frac{(1+ \abs{k})}{\abs{k} ( 1 + k^2)} 
 \norm{\left( N^{-1/2} a_k - \varphi_t(k) \right) U_K \Psi_{N,t}}
\nonumber \\
&\leq C \norm{\psi_t}_{H^1}  \left( \beta_K^b(t) + 
\int_{\abs{k} \geq K} d^3k \,   \frac{(1+ \abs{k})^2}{\abs{k}^2 ( 1 + k^2)^2}  \right)
\nonumber \\
&\leq C \norm{\psi_t}_{H^1}  \left( \beta_K^b(t) +  K^{-1} \right) .
\end{align}

\noindent \textbf{Line \eqref{eq: time derivative beta-b  3}.} Writing \eqref{eq: time derivative beta-b  3} as
\begin{align}  
2 \int_{\abs{k} \leq K} d^3k \, \abs{k}^{-1} \Im \scp{e^{ikx_1} \left( N^{-1/2} a_k - \varphi_t(k) \right) p_1^t U_K \Psi_{N,t}}{q_1^t U_K \Psi_{N,t}}
\end{align}
shows that this is exactly the same expression as the second line in \eqref{EST:TD:BETA:A:a:SECOND:LINE}.
We consequently have 
\begin{align}
\abs{\eqref{eq: time derivative beta-b  3}} 
&\leq C \norm{\psi_t}_{H^1}^2 \left( \beta^a_K (t) + \beta^b_K (t) \right).
\end{align}
\noindent \textbf{Line \eqref{eq: time derivative beta-b  4}.} To find a suitable bound for \eqref{eq: time derivative beta-b  4} is the most difficult step in the proof.
We start by estimating 
\allowdisplaybreaks
\begin{align}
\abs{ \eqref{eq: time derivative beta-b  4} }
&\leq
2 \int_{\abs{k} \leq K} 
d^3k \, \abs{k}^{-1} \Big| \scp{U_K \Psi_{N,t}}{e^{ikx_1} q_1^t \left( N^{-1/2} a_k - \varphi_t(k) \right) U_K \Psi_{N,t}} \Big|
\nonumber \\
& = 2 \int_{\abs{k} \leq K} d^3k \,  \abs{k}^{-1} 
\Big| \scp{U_K \Psi_{N,t}}{N^{-1} \sum_{j=1}^N e^{ikx_j} q_j^t \left( N^{-1/2} a_k - \varphi_t(k) \right) U_K \Psi_{N,t}} \Big|
\nonumber \\
&\leq 2 \int_{\abs{k} \leq K} d^3k \, \abs{k}^{-1} 
\Big\| N^{-1} \sum_{j=1}^N q_j^t e^{-ikx_j} U_K \Psi_{N,t} \Big\|
\norm{\left( N^{-1/2} a_k - \varphi_t(k) \right) U_K \Psi_{N,t}}
\nonumber \\
&\leq  \beta^b_K(t)   + \int_{\abs{k} \leq K} d^3k \, \abs{k}^{-2} 
\Big\| N^{-1} \sum_{j=1}^N q_j^t e^{-ikx_j} U_K \Psi_{N,t} \Big\|^2 .
\end{align}
The last term is bounded by 
\begin{align}
&\int_{\abs{k} \leq K}  d^3k \, \abs{k}^{-2} 
\Big\| N^{-1} \sum_{j=1}^N q_j^t e^{-ikx_j} U_K  \Psi_{N,t} \Big\|^2
\nonumber \\
&\quad \le  4 \pi N^{-1} K  +  \int_{\abs{k} \leq K} d^3k \, \abs{k}^{-2} 
\abs{ \Re \, \scp{q_2^t e^{-ikx_2} U_K \Psi_{N,t}}{q_1^t e^{- ik x_1} U_K  \Psi_{N,t}}  } .
\end{align}
With
 $O_{1,2} = \left( 1- \Delta_2 \right)^{-1/2} e^{ikx_2} \left( 1 - \Delta_1 \right)^{1/2} q_2^t $ one
has
\begin{align}
e^{ ikx_2} q_2^t q_1^t e^{- ikx_1} 
+ e^{ ikx_1} q_1^t q_2^t e^{- ikx_2} 
&= O_{2,1}^* O_{1,2}
+ O_{1,2}^* O_{2,1}
\leq O_{1,2}^* O_{1,2} + O_{2,1}^* O_{2,1}.
\end{align}
Thus, using the symmetry of the wave function and  $e^{-ikx_2} ( 1 - \Delta_2 )^{-1} e^{ikx_2} = ( ( - i \nabla_2 +  k )^2 + 1 )^{-1}$, we obtain the bound 
\begin{align}
&\int_{\abs{k} \leq K} d^3k \, \abs{k}^{-2} 
\abs{ \Re \, \scp{q_2^t e^{-ikx_2} U_K \Psi_{N,t}}{q_1 e^{- ik x_1} U_K  \Psi_{N,t}}  }
\nonumber \\
&\quad \leq 
\int_{\abs{k} \leq K} d^3k \, \abs{k}^{-2} 
\scp{U_K \Psi_{N,t}}{q_2^t \left( 1 - \Delta_1 \right)^{1/2} e^{- ikx_2} \left( 1 - \Delta_2 \right)^{-1} e^{ ikx_2} \left( 1 - \Delta_1 \right)^{1/2} q_2^t U_K \Psi_{N,t}}
\nonumber \\
&\quad =
\scp{  \left( 1 - \Delta_1 \right)^{1/2} q_2^t U_K \Psi_{N,t}}{ \int_{\abs{k} \leq K} d^3k \, \abs{k}^{-2}  \left( \left( - i \nabla_2 + k \right)^2 + 1 \right)^{-1}  \left( 1 - \Delta_1 \right)^{1/2} q_2^t U_K \Psi_{N,t}}.
\end{align}
In combination with
\begin{align}
\norm{\int_{\abs{k} \leq K} d^3k \, \abs{k}^{-2}  \left( \left( - i \nabla_2 + k \right)^2 + 1 \right)^{-1}}_{\op}
&= \sup_{p\in \mathbb R^3} \int_{\abs{k} \leq K} d^3k \, \abs{k}^{-2}  \left( \left( p + k \right)^2 + 1 \right)^{-1}
 < \infty
\end{align}
(compare with \eqref{eq: integral rearrangement}) and Lemma \ref{lemma: variance of the energy bound} this gives 
\begin{align}
\abs{ \eqref{eq: time derivative beta-b  4} }
& \leq C \left( \beta^b_K(t)  +  N^{-1} K +
\norm{(1 - \Delta_2)^{1/2} q_1^t U_K \Psi_{N,t}}^2
\right)
\nonumber \\
&\leq C \left( \beta^a_K (t) + \beta^b_K (t) + N^{-1} K + \norm{\nabla_2 q_1^t U_K \Psi_{N,t}}^2 \right)
\nonumber \\[2mm]
&\leq C \left( \norm{\psi_t}_{H^2}^2 
+ \norm{\varphi_t}_{L_1^2}^2  + \vert E_N^{\rm{F}}(\Psi_{N,t}) \vert \right)
\left(  \beta_K(t) + N^{-1} K^{-1} + N^{-1} K \right) .
\end{align}
\noindent \textbf{Line \eqref{eq: time derivative beta-b  5}.} We have 
\begin{align}
\abs{\eqref{eq: time derivative beta-b  5}}
&\leq 4 \int d^3k \,  \abs{k}  \abs{B_{K,x}(k)}
\norm{\left( N^{-1/2} a_k - \varphi_t(k) \right) U_K \Psi_{N,t}}   
\nonumber \\
&\qquad  \times  \norm{ \left( i \nabla_1 - N^{-1/2} \Phi \big( k B_{K,x_1} \big) \right) U_K \Psi_{N,t}}
\nonumber \\[1mm]
&\leq C  \sqrt{\beta^b_K(t)} \norm{ \left( i \nabla_1 - N^{-1/2} \Phi \big( k B_{K,x_1} \big) \right) U_K \Psi_{N,t}} 
\| \abs{\cdot}  B_{K,x} \| .
\end{align}
By using \eqref{eq: bounds for G}, the symmetry of the wave function and \eqref{eq: bound H^G vs H^0}, we get
\begin{align}
\abs{\eqref{eq: time derivative beta-b  5}}
&\leq C \left( \norm{\nabla_1 U_K \Psi_{N,t}} + \norm{\abs{\cdot} B_{K,x_1}} \norm{N^{-1/2} \left( \mathcal{N} +1 \right) U_K \Psi_{N,t}} \right) \left( \beta^b_K(t) + K^{-1} \right)
\nonumber \\[1mm]
&\leq C \left( \scp{U_K \Psi_{N,t}}{N^{-1} H_N^0  U_K \Psi_{N,t}}^{1/2} + 1 \right)
\left( \beta^b_K(t) + K^{-1} \right)
\nonumber \\[1mm]
&\leq C \left(    \scp{U_K \Psi_{N,t}}{ ( N^{-1} H_{N,K}^{\rm G} +C)  U_K \Psi_{N,t}}^{1/2} + 1 \right) \left( \beta^b_K(t) + K^{-1} \right).
\end{align}
In total, we thus arrive at 
\begin{align}\label{eq: summary time derivative beta b}
\abs{\frac{d}{dt} \beta^b_K(\Psi_{N,t}, \varphi_t)} 
&\leq C \left( \norm{\psi_t}_{H^2}^2 
+ \norm{\varphi_t}_{L_1^2}^2  + \vert E_N^{\rm{F}}(\Psi_{N,t}) \vert \right)
\left(  \beta_K(t) + K^{-1} + \frac{K}{N} \right). 
\end{align}
\noindent \textbf{Conclusion:}\ We combine $\frac{d}{dt}\beta^c(t) = 0$, \eqref{eq: summary time derivative beta a} and \eqref{eq: summary time derivative beta b} with Proposition \eqref{prop: solution theory for LPeq} and $\vert E_N^{\rm F}(\Psi_{N,t}) \vert = \vert E_N^{\rm F}(\Psi_N) \vert \le E_0$ in order to obtain \eqref{BOUND:BETA:TIME:DERIVATIVE}.\hfill$\square$

\section{Remaining proofs\label{sec: remaining proofs}}

\subsection{Proof of Lemma \ref{LEMMA:CLOSENESS:GAMMA}}
\label{SEC:PROOF:LEMMA:CLOSENESS:GAMMA}

To show the  inequality \eqref{BOUND:CLOSENESS:GAMMA:1:0}, we use $\gamma^{(1,0)}_{U_K\Psi_N} = \gamma^{(1,0)}_{U_{K,x_1}\Psi_N}$ with $U_{K,x_1} = \exp(i N^{-1/2} \Pi(B_{K ,x_1}))$, which follows directly from \eqref{eq: product structure Gross transform} and the definition of the reduced density matrix.
Hence, 
\begin{align}
\text{Tr}_{L^2(\mathbb R^3)} \Big\vert \gamma^{(1,0)}_{\Psi_N} - \gamma^{(1,0)}_{U_K\Psi_N}  \Big\vert  \le 2 \norm{(U_{K,x_1}-1) \Psi_{N}}.
\end{align}
Using $\norm{(U_{K,x_1}-1) \Psi_{N}} = \norm{(U_{K,x_1}-1) U_K \Psi_{N}}$, we obtain \eqref{BOUND:CLOSENESS:GAMMA:1:0} from the bound
\begin{align}\label{BOUND:STRONG:CONV:GROSS:TRAFO}
\norm{ ( U_{K,x_1} -1)  U_K \Psi_N } \le 2   \norm{B_{K,x}} \, \norm{\Big(\frac{\mathcal N+1}{N} \Big)^{1/2} U_K \Psi_N }
\end{align}
together with $\mathcal N \le H_N^0$, \eqref{eq: bound H^G vs H^0} and \eqref{REPRESENTATION:H:F}. Inequality \eqref{BOUND:STRONG:CONV:GROSS:TRAFO} follows from the spectral calculus for self-adjoint operators, using $1-U_{K,x}=f(N^{-1/2}\Pi(B_{K,x}))$ with $f(s) = 1-\exp(i s)$ in combination with $\vert f(s) \vert \le \vert s \vert $.

Using the properties of the Weyl operator  (in particular \eqref{eq: Weyl operators product} together with \eqref{eq: product structure Gross transform}) and
\begin{align}
U_K\mathcal{N} U^*_K 
&= \mathcal{N} + N^{-1} \sum_{i,j=1}^N \scp{B_{K,x_i}}{B_{K,x_j}}
+ N^{-1/2} \sum_{j=1}^N \left( a  ( B_{K, x_j}  ) + a^*  ( B_{K, x_j}  ) \right),
\end{align}
we have
\begin{align} \nonumber
&N^{-1} \abs{  \scp{ W^*(\sqrt N \varphi) \Psi_{N} }{ \big( \mathcal N - U_K^* \mathcal N U_K \big)  W^* (\sqrt N \varphi)  \Psi_N } } \\[1.5mm]
&\quad =
  N^{-1} \Big| \scp{ W^*(\sqrt N \varphi) U_K \Psi_{N} }{\big( U_K \mathcal N U^*_K -\mathcal N \big) W^*(\sqrt N \varphi) U_K \Psi_N }  \nonumber \\
&\quad\leq \norm{B_{K,x}}^2
+ 2 N^{-3/2} \sum_{j=1}^N \norm{a\left( B_{K,x_j} \right) W^*(\sqrt{N} \varphi) U_K \Psi_N}
\nonumber \\
&\quad \leq 
  \norm{B_{K,x}}^2 
+ 2    \norm{B_{K,x}} \norm{\varphi} 
+ 2 N^{-1/2}   \norm{B_{K,x}} \norm{\mathcal{N}^{1/2} U_K \Psi_N} .
\label{eq:  bound for the difference of the coherent states calculation}
\end{align}
An application of \eqref{eq: bounds for G} and \eqref{eq: bound H^G vs H^0} then leads to
\begin{align}
\eqref{eq:  bound for the difference of the coherent states calculation}
&\leq C  (1+\norm{\varphi})   \norm{B_{K,x}}
\big( 1 + \scp{U_K \Psi_N}{N^{-1} H_N^0 U_K \Psi_N}^{1/2} \big)
\nonumber \\[2mm]
&\leq C   K^{-3/2}  ( 1+\norm{\varphi}) \scp{U_K \Psi_N}{ ( N^{-1} H_{N,K}^{\rm G} + C ) U_K \Psi_N }^{1/2}.
\end{align}
In combination with \eqref{REPRESENTATION:H:F}, this shows \eqref{BOUND:CLOSENESS:GAMMA:0:1}.
\hfill$\square$

\subsection{Proof of Proposition \ref{proposition: initial bounds for gross transformed product} \label{sec: initial states}}

Throughout this section, we set $\xi_N = \psi^{\otimes N}\otimes W(\sqrt N \varphi) \Omega$.
The bound on the energy follows from
\begin{align}
\scp{\Psi_{N}}{H_N^{\rm F} \Psi_N} = \scp{\xi_{N}}{H_{N,K}^{\rm G}  \xi_{N} } 
= \scp{\xi_{N}}{ \Big( H_{N,K}^{\rm F} + \sum_{j=1}^N A_{K,x_j}  + \sum_{j,l=1}^N V_{K}(x_j-x_l) \Big)  \xi_{N}}
\end{align}
in combination with \eqref{eq: bound for V_K}, \eqref{eq: bound for A_K} and
\begin{align}\label{eq: bound for exp val wrt xi}
 \abs{ \scp{\xi_{N}}{ H_{N,K}^{\rm F}   \xi_{N}} }
 & = N \abs{ \scp{\psi}{\left( - \Delta + \Phi_{K}(\cdot,0) \right) \psi} + \norm{\varphi}^2  } \leq C N \left( \norm{\psi}_{H^1}^2 + \norm{\varphi}^2 \right).
\end{align}
In \eqref{eq: bound for exp val wrt xi}, we used the shift property of the Weyl operators \eqref{eq: Weyl operators shift property}.

For the bound on $a(\Psi_{N},\psi)$, we note that 
\begin{align}
\tr_{L^2(\mathbb{R}^3)} \abs{\gamma^{(1,0)}_{\Psi_{N}} - \ket{\psi} \bra{\psi}} = \tr_{L^2(\mathbb{R}^3)} \abs{\gamma^{(1,0)}_{\Psi_{N}} - \gamma^{(1,0)}_{ U_K \Psi_{N}}  }
\end{align}
since $\ket{\psi} \bra{\psi} = \gamma^{(1,0)}_{\xi_{N}}$ and $\xi_{N} = U_K\Psi_{N}$. Applying Lemma \eqref{LEMMA:CLOSENESS:GAMMA}, we obtain the stated estimate. 
The bound on $b(\Psi_{N},\varphi)$ follows readily from 
\eqref{eq: Gross transform shift property},
\begin{align}
b(\Psi_{N},\varphi ) &= N^{-1} \int d^3k \, \norm{a_k \, U^{*}_K ( \psi^{\otimes N}  \otimes  \Omega ) }^2
\leq    \norm{B_{K,x}}^2 \leq C K^{-3} .
\end{align}
We are thus left with the bound for $c(\Psi_{N})$, which we write with the aid of  \eqref{REPRESENTATION:H:F} as
\begin{align}
c(\Psi_{N})  = \norm{ N^{-1} \left( H_{N,K}^{\rm G} - \scp{\xi_{N}}{ H_{N,K}^{\rm G} \xi_{N}} \right) \xi_{N}}^2.
\end{align}
Recalling Definition \ref{definition: Gross transformed Froehlich Hamiltonian} and using the triangle inequality, we get
\begin{align}
\norm{ \left( H_{N,K}^{\rm G} - \scp{\xi_{N}}{ H_{N,K}^{\rm G} \xi_{N}} \right) \xi_{N}} &\le   \norm{\left( H_{N,K}^{\rm F} - \scp{\xi _{N}}{ H_{N,K}^{\rm F} \xi_{N}} \right) \xi_{N}} 
\nonumber \\[2mm]
&  + 2   \norm{  \sum_{j=1}^N A_{K,x_j}  \xi_{N}}  + 2 \norm{\sum_{j,l=1}^NV_K(x_j-x_l) \xi_N} .
\end{align}
After a lengthy but straightforward computation, using the shift property \eqref{eq: Weyl operators shift property} and the fact that $\Delta_x$ commutes with $W(\sqrt N \varphi)$, we find that  
\allowdisplaybreaks
\begin{align}
\label{eq: variance of H_K computation}
& \frac 1 N  \norm{ \left( H_{N,K}^{\rm F} - \scp{\xi_{N}}{  H_{N,K}^{\rm F} \xi_{N}} \right) \xi_{N}}^2
\nonumber \\[2mm]
&  =  \norm{\varphi}^2 + \scp{\psi}{(-\Delta)^2\psi} - \scp{\psi}{(-\Delta)\psi}^2 + N^{-1} \scp{\psi}{\norm{G_{K,x}}^2\psi}  \nonumber \\[2mm]
& +(1-N^{-1}) \norm{ \scp{\psi}{G_{K,x}\psi} }^2  + 4 \scp{\psi}{ (\Re\scp{G_{K,x}}{\varphi})^2\psi} - 4  \scp{\psi}{ \Re  \scp{G_{K,x}}{\varphi}\psi}^2  \nonumber \\[4mm]
& +2  \Re \scp{\psi}{\scp{\varphi}{G_{K,x}}\psi}   + 2 \big( \scp{\psi}{(-\Delta_x)\Re \scp{G_{K,x}}{\varphi} \psi} + \text{c.c.}  \big)
- 4 \scp{\psi}{(-\Delta)\psi} \scp{\psi}{\Re \scp{G_{K,x}}{\varphi} \psi} .
\end{align}
We shall show that the right hand side is bounded from above by a constant times $1+K N^{-1}$,  with the constant depending only on $\norm{\psi}_{H^2}$ and $\norm{\varphi}_{L^2_1}$. For the first four summands (i.e., the terms in the first line), this is obvious (recall \eqref{eq: bounds for G}). In the fifth summand, we can use  \eqref{eq: Lieb Yamazaki for expectation value of e-ikx} to conclude that 
$\norm{ \scp{\psi}{G_{K,x}\psi} }^2\le C$ independently of $K$. For each of the remaining terms on the right side of \eqref{eq: variance of H_K computation}, we use
\begin{equation}
\abs{\scp{G_{K,x}}{\varphi}}    \le  \norm{ (1+ \abs{\,\cdot \, })^{-1} G_{K,x} }_2 \, \norm{ (1+ \abs{\,\cdot\,}) \varphi},
\end{equation}
which is bounded by $C \norm{\varphi}_{L_1^2}$. Hence, we find
\begin{align}
\norm{N^{-1} \left( H_{N,K}^{\rm F} - \scp{\xi_{N}}{ H_{N,K}^{\rm F} \xi_{N}} \right) \xi_{N}}^2    \le C\big( N^{-1} + KN^{-2} \big) .
\end{align}
Next, we use
\begin{align}
\frac{1}{N} \norm{\sum_{j=1}^N A_{K,x_j} \xi_N} \le \norm{A_{K,x_1}\xi_N},
\end{align}
and recalling \eqref{eq: Definition of V-K-N}  we estimate (with $\widetilde \Psi_N = \psi^{\otimes N}\otimes \Omega$)
\begin{align}
 & \norm{  N^{-1/2}\nabla_1 \cdot
 a \big( k B_{K,x_1} \big)   W (\sqrt{N} \varphi)  \widetilde \Psi_N}  \nonumber \\[2mm]
 & =   \norm{  \nabla_1 \cdot
 \int dk\, k \overline{B_{K,x_1}(k) }   \varphi(k) \widetilde \Psi_N} \nonumber \\[2mm]
  & \le    \norm{  \int dk \, k^2  \overline{B_{K,x_1}(k) }  \varphi(k) \widetilde \Psi_N} +    \norm{  \int dk\, k  \overline{B_{K,x_1}(k) } \varphi(k) \nabla_1 \widetilde \Psi_N} \nonumber \\[4mm]
  & \le  \norm{ \abs{\cdot}  B_{K, x } }_2  \big( \norm{\abs{\cdot}  \varphi }  +   \norm{\varphi} \norm{\psi}_{H^1} \big) \le CK^{-1/2}.
\end{align}
Similarly also
\begin{align}
 \norm{  N^{-1/2} 
 a^* \big( k B_{K,x_1} \big) \cdot \nabla_1   W (\sqrt{N} \varphi)  \widetilde \Psi_N} & =  \norm{    
 \int dk \,k \, B_{K_N,x_1}(k) \big( N^{-1/2}a_k^* + \overline{\varphi(k)}\big) \cdot \nabla_1 \widetilde \Psi_N} \nonumber \\[3mm]
& \le  \norm{ \abs{\cdot}  B_{K,x}} \big( N^{-1/2}  \norm{\psi}_{H^1} +  \norm{\varphi } \norm{\nabla \psi}\big) \le C  K^{-1/2} .
\end{align}
In order to estimate the term containing $\Phi(kB_{K,x})^2$, consider
\begin{align}
W^*(\sqrt N \varphi) \Phi(k B_{K ,x_1}) W(\sqrt N \varphi) & = \Phi(k B_{K ,x_1}) + 2 \sqrt N\Re \scp{kB_{K ,x_1}}{\varphi},
\end{align}
and thus
\begin{align}
& \norm{ N^{-1 }\Phi(k B_{K,x_1})^2  W(\sqrt{N}\varphi) \widetilde \Psi_N} 
\nonumber \\[2mm]
& =  N^{-1} \norm{\Big( \Phi(k B_{K ,x_1}) + 2 \sqrt N\Re \scp{ k B_{K ,x_1}} {\varphi}\Big)^2 \widetilde \Psi_N} 
\nonumber \\[2mm] 
& \le 2 N^{-1 }\norm{ \Phi(k B_{K ,x_1})^2  \widetilde \Psi_N} 
  + 8 \vert \scp{  {|} k B_{K  ,x_1}  {|} } { {|} \varphi  {|}} \vert^2.
\end{align}
In the last line, we use  \eqref{eq: bounds for G} to obtain  
\begin{equation}
\norm{ \Phi(k B_{K, {x_1}})^2  \widetilde \Psi_N} = \sqrt{3}   \norm{k B_{K , {x}}}^2 \leq C K^{-1}. 
\end{equation}
Finally, using \eqref{eq: bound for V_K}, we estimate 
\begin{align}
N^{-1} \norm {\sum_{j,l=1}^N V_K(x_j-x_l) \xi_N } \le N^{-1} \sum_{j,l=1}^N\norm{V_K(x_j-x_l) \xi_N} \le C K^{-1},
\end{align}
which completes the proof of the proposition.
\hfill$\square$

\subsection{Proof of Theorem \ref{theorem: main theorem: 2}}
\label{section: proof of the proposition for free domain}

Given Theorem \ref{theorem: main theorem}, \eqref{MAIN:BOUND:TRACE:NORM:1:0:THEOREM:2} follows from \eqref{BOUND:CLOSENESS:GAMMA:1:0} 
 together with the bound
\begin{align}\label{eq: strong convergence of U}
\norm{(1-U_K) \Psi_{N}} \le \frac{C N}{ K^{3/2} } \norm{\Big( \frac{H_{N,K}^{\rm G} + C N }{N}\Big)^{1/2} \Psi_{N}},\quad \Psi_N \in \mathcal{D} ( H_N^0 ).
\end{align}
The latter follows from $\mathcal N \le H^0_N$, \eqref{eq: bound H^G vs H^0} and the functional calculus for self-adjoint operators, using $1-U_K = f(  N^{-1/2} \sum_{j=1}^N \Pi(B_{K,x_j} ) )$ with $f(s) = 1-\exp(is)$ and the bound $\vert f(s) \vert \le \vert s \vert$. In more detail, let $\Psi_{N,t}$ as in Theorem \ref{theorem: main theorem: 2} and denote $\Phi_{N,t} = e^{- i H_N^{\rm F} t} U^*_K \Psi_{N,0}$. Then, using \eqref{eq: strong convergence of U},
\begin{align}
\textnormal{Tr}_{L^2(\mathbb{R}^3)} \abs{ \gamma^{(1,0)}_{\Psi_{N,t}} - \gamma^{(1,0)}_{\Phi_{N,t} } } \le 2 \norm{e^{-iH_N^{\rm F} t } (\Psi_{N,0} - \Phi_{N,0} )} = 2 \norm{(1-U_K) \Psi_{N,0}} \le \frac{CN}{K^{3/2}}, 
\end{align}
and the triangle inequality,
\begin{align}
\textnormal{Tr}_{L^2(\mathbb{R}^3)} \abs{\gamma^{(1,0)}_{\Psi_{N,t}} - \ket{\psi_t} \bra{\psi_t} } \le \frac{C N}{K^{3/2}}  + \textnormal{Tr}_{L^2(\mathbb{R}^3)} \abs{\gamma^{(1,0)}_{\Phi_{N,t} } - \ket{\psi_t} \bra{\psi_t} }.
\end{align}
Since $\Phi_{N,0} \in \mathcal{D} ( H_N^{\rm F} )$ and 
$E_0 = \sup_{N \in \mathbb{N}} \vert N^{-1} \scp{\Phi_{N,0}}{H_N^{\rm F} \Phi_{N,0}} \vert < \infty$ by assumption, we infer with Theorem \ref{theorem: main theorem} that
\begin{align}
& \textnormal{Tr}_{L^2(\mathbb{R}^3)} \abs{\gamma^{(1,0)}_{\Phi_{N,t}} - \ket{\psi_t} \bra{\psi_t} }  \nonumber \\[1mm]
&\quad  \quad \le  \sqrt{a (\Phi_{N,0 },\psi) +  b (\Phi_{N,0},\varphi) + c (\Phi_{N,0}) + N^{-1/2} }  e^{C(1 + \abs{t})^3}.
\end{align}
Using Lemma \ref{LEMMA:CLOSENESS:GAMMA}, we have 
\begin{align}\label{eq: bounds a,b,c}
  a(\Phi_{N,0},\psi) \leq a(\Psi_{N,0 },\psi) + C K^{-3/2} , \quad b(\Phi_{N,0 },\varphi ) \leq  b(\Psi_{N,0} ,\varphi )  + CK^{-3/2},
\end{align}
which proves the first bound in Theorem \ref{theorem: main theorem: 2} if we set $K=K_N\ge c N^{5/6}$.

In order to prove \eqref{MAIN:BOUND:TRACE:NORM:0:1:THEOREM:2}, we estimate
\begin{subequations}
\begin{align}
& \scp{W^*  (\sqrt{N} \varphi_t) \Psi_{N,t}}{\sqrt{ \tfrac{\mathcal N}{N} } \, W^* (\sqrt{N} \varphi_t) \Psi_{N,t}} \nonumber \\[2mm]
& \quad \quad \quad \le \big\vert \scp{W^{*} (\sqrt{N} \varphi_t) \Psi_{N,t}}{ \sqrt{ \tfrac{\mathcal N}{N} }  \, W^{*} (\sqrt{N} \varphi_t) \Phi_{N,t}} \big\vert \\[1mm] 
& \quad \quad \quad + \big\vert \scp{W^{*} (\sqrt{N} \varphi_t)  \Psi_{N,t}}{\sqrt{ \tfrac{\mathcal N}{N} }  \, W^{*} (\sqrt{N} \varphi_t)  (\Psi_{N,t}-\Phi_{N,t} )  } \big\vert \label{eq: proof of sqrt N line 2}
\end{align}
\end{subequations}
with $\Phi_{N,t}$ defined as above. In the first line, we use the Cauchy--Schwarz inequality and apply Theorem \ref{theorem: main theorem} to $\Phi_{N,t}$, i.e.,
\begin{align}
\norm{\sqrt{ \frac{\mathcal N}{N} } W^*(\sqrt N \varphi_t) \Phi_{N,t} }^2 & \le \Big( a (\Phi_{N,0},\psi) +  b (\Phi_{N,0},\varphi) + c(\Phi_{N,0}) + N^{-1/2} \Big) e^{C(1 + \abs{t})^3} 
\nonumber \\
& \le \Big( a (\Psi_{N,0},\psi) +  b (\Psi_{N,0},\varphi) + c(\Phi_{N,0}) +  K^{-3/2} +  N^{-1/2} \Big)   e^{C(1 + \abs{t})^3},
\end{align}
where we made use of \eqref{eq: bounds a,b,c} in the second step. In  \eqref{eq: proof of sqrt N line 2}, we estimate
\begin{align}
\norm{\Psi_{N,t} - \Phi_{N,t}} \le \norm{(1-U_{K}) \Psi_{N,0}} \le \frac{C N}{ K^{3/2} } \norm{\Big( \frac{H_{N,K}^{\rm G} +C  N }{N}\Big)^{1/2} \Psi_{N,0}},
\end{align}
together with
\begin{align}\label{eq: bound sqrt N}
\norm{\sqrt{ \frac{\mathcal N}{N} } W^*(\sqrt N \varphi) \Psi_{N,t}} \le C\Bigg(  \norm{\Big( \frac{H_{N,K}^{\rm G} +C N}{N} \Big)^{1/2} \Psi_{N,0}} + \norm{\varphi_t}\Bigg),
\end{align}
which for $K=K_N\ge c N^{5/6}$ proves \eqref{MAIN:BOUND:TRACE:NORM:0:1:THEOREM:2}. In order to show \eqref{eq: bound sqrt N}, we use the commutation relations \eqref{eq: Weyl operators shift property} and $2 \Phi(\sqrt N \varphi_t) \le \mathcal N + N\norm{\varphi_t}^2$, in order to find
\begin{align}
\scp{W^*(\sqrt N \varphi_t) \Psi_{N,t}}{ \mathcal N W^*(\sqrt N \varphi_t) \Psi_{N,t}} \le 2 \big( \scp{ \Psi_{N,t}}{ \mathcal N \Psi_{N,t}} + N \norm{ \varphi_t }^2  \big).
\end{align}
Using \eqref{eq: bound H^F vs H^0} in combination with \eqref{eq: bound H^G vs H^0} leads to
\begin{align}
\mathcal N \le  H_N^0   \le  2 H_N^{\rm F} + C N  & =  e^{- iH_N^{ \rm F} t} (2  H_N^{\rm F} + C N)  e^{i H_N^{\rm F} t }  \le C e^{- iH_N^{ \rm F} t} ( H_{N,K}^{\rm G} + C N)  e^{i H_N^{\rm F} t }.
\end{align}
It remains to show \eqref{eq: redyced density bound  product state} and \eqref{eq: redyced density bound  product state line 2}: For the Pekar state $\Psi_{N} = \psi^{\otimes N} \otimes W(\sqrt{N} \varphi) \Omega$, we have 
\begin{align}
 a(\Psi_{N},\psi ) = 0, \quad  b(\Psi_{N},\varphi ) =0, \quad c( U_{K}^* \Psi_{N} ) \leq  C \Big(N^{-1} + K^{-1} + \frac{K}{N^2}  \Big),
\end{align}
where the last bound was proven in Proposition \ref{proposition: initial bounds for gross transformed product}. Thus, if we choose $K= K_N = c N$, we obtain \eqref{eq: redyced density bound  product state} and \eqref{eq: redyced density bound  product state line 2}.
\hfill$\square$

\appendix

\section{Auxiliary bounds}\label{appendix: interaction bounds}

In this appendix, we collect bounds on the interaction terms of the Hamiltonians $H_N^{\rm F}$ and $H_{N,K}^{\rm G}$ and derive the frequently used inequalities \eqref{eq: bound H^F vs H^0} and \eqref{eq: bound H^G vs H^0}. After that we comment on the proof of Lemma \ref{lemma representation of H_F}.

\begin{lemma}
\label{lemma: bounds for the interaction term}
For every $\varepsilon >0$,  $K \in (0,\infty]$, $N \in \mathbb{N}$ and $j \in \{1,\ldots,N\}$, we have 
\begin{align}\label{eq: bound interaction vs free Ham}
\pm N^{-1/2} \Phi(G_{K,x_j}) 
& \leq \varepsilon \left( - \Delta_j + \frac{  \mathcal{N}+1}N  \right) +2  \frac {(16\pi)^2}{\varepsilon^3} 
\end{align}
on $L^2(\mathbb R^{3N}) \otimes \mathcal F_s$. 
Moreover, with $A_{K,x}$  defined  in Definition \ref{definition: Gross transformed Froehlich Hamiltonian}, 
\begin{align}
\pm A_{K,x_j} &\leq \sqrt{\frac{64  \pi }{K}} \left( - \Delta_j + N^{-1} \mathcal{N} \right) + \frac{16\pi}{N K}
\label{eq: bound for A_K}.
\end{align}
\end{lemma}

\begin{proof} To prove \eqref{eq: bound interaction vs free Ham}, we use again the commutator method by Lieb and Yamazaki \cite{liebyamazaki}. Using \eqref{eq: bound for Pi} and \eqref{eq: bounds for G}, we have 
\begin{align}\label{eq: Phi bound appendix}
\abs{ \scp{\Psi_N}{N^{-1/2}\Phi(G_{K,x_j})   \Psi_N} } 
& \le \frac{4\pi K  }{\varepsilon} \norm{\Psi_N}^2 + {\varepsilon}N^{-1}  \scp{\Psi_N}{(\mathcal N+1) \Psi_N},
\end{align}
which proves \eqref{eq: bound interaction vs free Ham} for  $K\le 64 \pi/\varepsilon^2$. In the case $K>64 \pi/\varepsilon^2$, we write $\Phi (G_{K,x_j}) =   \Phi (G_{K',x_j})  + (  \Phi(G_{K,x_j})  -  \Phi(G_{K',x_j}) )$ with $K' = 16 \pi/\varepsilon^2 $. For the first summand, we use \eqref{eq: Phi bound appendix} with $K$ replaced by $K'$ and $\varepsilon$ replaced by $\varepsilon/2$,
while for the remainder, write
\begin{equation}
  \scp{\Psi_N}{( \Phi(G_{K,x_j}) -\Phi(G_{K',x_j})  )\Psi_N} 
 =   \scp{\Psi_N}{ [\nabla_j, \Phi(g_{x_j}) ] 	 \Psi_N} 
 \end{equation}
with $g_x(k) = i k \vert k \vert^{-3}e^{-ikx} \id_{K'\le \vert k \vert \le K}(k)$.
The absolute value of the last expression is bounded from above by
\begin{equation}
4   \norm{\nabla_j \Psi_N} \| g_x\| \norm{\sqrt{ \mathcal N+1}  \Psi_N}  \le 2  N^{1/2} \norm{g_x} \scp{\Psi_N}{\Big(-\Delta_j + \frac{\mathcal N + 1}{N} \Big) \Psi_N }.
\end{equation}
Using $\norm{g_x}\le \sqrt{4\pi / K'} = \frac{\varepsilon}{4}$ shows \eqref{eq: bound interaction vs free Ham}. 

To show \eqref{eq: bound for A_K}, we use $\norm{\abs{\cdot } B_{K,x_j}}^2 \leq 4 \pi K^{-1}$ and 
\begin{align}
& \hspace{-0.25cm}\abs{\scp{\Psi_N}{A_{K,x_j} \Psi_N}}\nonumber \\[2mm]
&\leq 4 N^{-1/2} \abs{\scp{\nabla_j \Psi_N}{a \big( k B_{K,x_j} \big) \Psi_N}} + N^{-1} \norm{\Phi \big( k B_{K,x_j} \big) \Psi_N }^2
\nonumber \\[2mm]
&\leq 4 N^{-1/2} \norm{\nabla_j \Psi_N} \norm{\abs{\cdot } B_{K,x}} \norm{\mathcal{N}^{1/2} \Psi_N}   + 4 N^{-1}  \norm{\abs{\cdot} B_{K,x}}^2 \norm{\left( \mathcal{N} + 1 \right)^{1/2} \Psi_N}^2
\nonumber \\[0.5mm]
&\leq \sqrt{\frac{16 \pi}{K}} \scp{\Psi_N}{\left( - \Delta_j + N^{-1} \mathcal{N}  \right) \Psi_N} + \frac{16 \pi}{K}   \scp{\Psi_N}{\frac{\mathcal N + 1}{N}\Psi_N}  .
\end{align}
\end{proof}

The previous lemma readily implies the validity of the bounds \eqref{eq: bound H^F vs H^0} and \eqref{eq: bound H^G vs H^0}.
With 
\begin{align}
H_N^{\rm F} = H_N^0 + N^{-1/2}\sum_{j=1}^N \Phi (G_{\infty,x_j}),
\end{align}
we can use \eqref{eq: bound interaction vs free Ham} with $\varepsilon = 1/2$ in order to infer \eqref{eq: bound H^F vs H^0}. 
Using in addition \eqref{eq: bound for V_K} and \eqref{eq: bound for A_K}, one similarly obtains \eqref{eq: bound H^G vs H^0}.\hfill $\square$\medskip

\noindent \textbf{Comment on the proof of Lemma \ref{lemma representation of H_F}.} As  already explained, Lemma \ref{lemma representation of H_F} was stated and proved in \cite{griesemerwuensch} for the case $N=1$. Since the statement $N\ge 2$ can be proven by almost literal adaption of the argument from \cite{griesemerwuensch} (with obvious minor modifications), we omit all details except for the proof of the following lemma. The bound given in the lemma is one of the main ingredients in the proof, and in particular its $N$-dependence is crucial since it guarantees that we can choose $\widetilde K$ in Lemma \ref{lemma representation of H_F} independently of $N$.

\begin{lemma}
\label{Lemma: norm bound of the Gross transformed Hamiltonian} For any $\varepsilon >0$ there are $K_{\varepsilon} >0$ and $C_{\varepsilon} >0$ such that for all $N \in \mathbb{N}$, $K \geq K_{\varepsilon}$ and any $\Psi_N \in \mathcal{D} ( H_N^0 )$,
\begin{align}
\label{eq: norm bound of the Gross transformed Hamiltonian}
\norm{\left( H_{N,K}^{\rm G} - H_N^0 \right) \Psi_N} \leq \varepsilon \norm{H_N^0 \Psi_N} + C_\varepsilon K N \norm{\Psi_N}.
\end{align}
\end{lemma}
\begin{proof}
We estimate each term in
\begin{align}
H_{N,K}^{\rm G} - H_N^0 &= \sum_{j=1}^N \left( N^{-1/2} \Phi(G_{K,x_j})   + A_{K,x_j} \right) 
+ \sum_{j,l=1}^N V_{K}(x_j-x_l)  
\end{align}
separately. Using \eqref{eq: bound for Pi}, $\norm{G_{K,x}} \leq C \sqrt{K}$ and $\norm{\mathcal{N}^{1/2} \Psi_N}^2 \leq \norm{H_N^0 \Psi_N} \norm{\Psi_N}$,
we have
\begin{align}
N^{-1/2} \sum_{j=1}^N \norm{ \Phi \left( G_{K,x_j} \right) \Psi_N}
&\leq C \sqrt{K N} \left( \norm{\mathcal{N}^{1/2} \Psi_N} + \norm{\Psi_N} \right)
\nonumber \\
&\leq \delta \norm{H_N^0 \Psi_N} + C K N \left( 1 + \delta^{-1} \right) \norm{\Psi_N}  
\end{align}
for any $\delta >0$. The norm of $\sum_{j,l =1}^N V_{K}(x_j-x_l) \Psi_N$ can be bounded using \eqref{eq: bound for V_K}. From \eqref{eq: Definition of V-K-N} we see that the remaining terms to estimate are the following:
\begin{align}
N^{-1} \sum_{j=1}^N \norm{\Phi \left( k B_{K, x_j} \right)^2 \Psi_N}
&\leq C K^{-1} \norm{\left( \mathcal{N} + 1 \right) \Psi_N}
\nonumber \\
&\leq C K^{-1}  \left(  \norm{H_N^0 \Psi_N} + \norm{\Psi_N}  \right),
\end{align}
where we have used \eqref{eq: bound for Pi} and \eqref{eq: bounds for G}. 
Similarly, we have
\begin{align}
2 N^{-1/2} \sum_{j=1}^N \norm{a^* \left( k B_{K, x_j} \right) \nabla_j \Psi_N }
&\leq C \sqrt{N K^{-1}} \norm{\sqrt{\mathcal{N} + 1} \nabla_j \Psi_N}
\nonumber \\
&\leq C K^{-1/2} \left(  \norm{H_N^0 \Psi_N} + \norm{\Psi_N}  \right) 
\end{align}
and 
\begin{align}
2 N^{-1/2} \sum_{j=1}^N \norm{\nabla_j a ( k B_{K, x_j} ) \Psi_N }
&\leq 2 N^{-1/2} \sum_{j=1}^N \left(  \norm{ a ( k B_{K, x_j} ) \nabla_j  \Psi_N }
+ \norm{ a ( \abs{k}^2 B_{K, x_j} ) \Psi_N } \right)
\nonumber \\
&\leq C K^{-1/2} \norm{H_N^0 \Psi_N} + 2 N^{1/2} \norm{a ( \abs{k}^2 B_{K,x_1} ) \Psi_N}.
\end{align}
In order to estimate the second summand, we use
\begin{align}\label{eq: a*a bound in appendix}
a^* ( \abs{k}^2 B_{K ,x_1} ) a ( \abs{k}^2 B_{K,x_1} )  \le \widetilde C_K (1-\Delta_1) \mathcal N ,
\end{align}
where  
\begin{align}
\label{eq: constant in appendix}
\widetilde{C}_K =  \sup_{h \in \mathbb{R}^3}  \int_{\abs{k} \geq K} d^3k \, \frac{ 1 }{ ( 1 + \abs{k}^2 ) \left( 1 +  ( h - k \right)^2 ) } .
\end{align}
The bound \eqref{eq: a*a bound in appendix} is analogous to \eqref{eq: bound for a*a analogous to Fran and Schlein} and can be proven in the same way as  \cite[Lemma~10]{frankschlein} (see also \cite[Lemma B.5]{griesemer}). If one estimates the integral in \eqref{eq: constant in appendix}  using the Cauchy--Schwarz inequality one sees that $\widetilde C_K\to 0$ for $K\to \infty$. We thus have
\begin{align}
2 N^{1/2} \norm{a ( \abs{k}^2 B_{K,x_j} ) \Psi_N}
&\leq 2 \widetilde{C}_K N^{1/2} \norm{\left( 1 - \Delta_j \right)^{1/2} \mathcal{N}^{1/2} \Psi_N}
\nonumber \\
&= 2  \widetilde{C}_K  \scp{\Psi_N}{  \mathcal{N} \sum\nolimits_{j=1}^N \left( 1 - \Delta_j \right) \Psi_N}^{1/2} 
\nonumber \\
&\leq  \widetilde{C}_K \left( \norm{H_N^0 \Psi_N} + N \norm{\Psi_N} \right),
\end{align}
and hence,
\begin{align}
2 N^{-1/2} \sum_{j=1}^N \norm{\nabla_j a \left( k B_{K , x_j} \right) \Psi_N }
&\leq  \left(  \widetilde{C}_K + C K^{-1/2} \right) \left( \norm{H_N^0 \Psi_N} +  N  \norm{\Psi_N} \right).
\end{align}
Choosing $K$ large enough and $\delta$ sufficiently small completes the proof of the lemma.
\end{proof}

\section{Time derivative of $\beta^b_K(t)$}

\label{section: rigorous derivation of the time derivative}

Because of the unboundedness of the annihilation operator, it is not directly obvious that one can use the product rule of differentiation to obtain \eqref{eq: time derivative of beta-b product rule a} and \eqref{eq: time derivative of beta-b product rule b}. Its rigorous justification relies on the estimate (for $\chi_N \in \mathcal{D}(H_N^0)$)
\begin{align}
\label{eq: pull through formula}
\left \|\mathcal{N} \chi_N \right \|
&\leq
\left \| H_N^0 \chi_N \right \|
\leq 2  \left \| H_{N,K}^{\rm{G}} \chi_N \right \| + C K N \left\| \chi_N \right \| ,
\end{align}
which follows from Lemma \ref{Lemma: norm bound of the Gross transformed Hamiltonian}. Since $U_K \Psi_{N,t} = e^{-i H_{N,K}^{\rm{G}} t} U_K  \Psi_{N,0}$, this together with the strong continuity of $e^{-i H_{N,K}^{\rm{G}} t}$ implies
\begin{align}
\label{eq: convergence with number operator}
\lim_{h \rightarrow 0} \left \| \left( \mathcal{N} + 1 \right) U_K \left(  \Psi_{N,t+h} -  \Psi_{N,t} \right) \right \| = 0
\quad \text{for all} \;  U_K  \Psi_{N,0} \in \mathcal{D} \left( H_N^0 \right).
\end{align}
Note that 
\begin{subequations}
\begin{align}
\label{eq: rigirous derivative beta-b A}
\beta_K^b(\Psi_{N,t+h}, \varphi_{t+h}) - \beta_K^b(\Psi_{N,t}, \varphi_t)
&=    \beta_K^b(\Psi_{N,t+h}, \varphi_t) -  \beta_K^b(\Psi_{N,t}, \varphi_t)
\\
\label{eq: rigirous derivative beta-b B}
&\quad +  \beta_K^b(\Psi_{N,t+h}, \varphi_{t+h}) - \beta_K^b(\Psi_{N,t+h}, \varphi_{t})
\end{align}
\end{subequations}
and that the first line is given by
\begin{align}
\eqref{eq: rigirous derivative beta-b A} &= 2  N^{-1} \Re \scp{W ( \sqrt{N} \varphi_t ) \mathcal{N} W^*(\sqrt{N} \varphi_t) U_K \Psi_{N,t}}{ U_K \left( \Psi_{N,t+h} - \Psi_{N,t} \right)}
\nonumber \\
&\quad + N^{-1} \scp{W(\sqrt{N} \varphi_t) \mathcal{N} W^*(\sqrt{N} \varphi_t )  U_K \left( \Psi_{N,t+h} - \Psi_{N,t} \right)}{ U_K \left( \Psi_{N,t+h} - \Psi_{N,t} \right)}.
\end{align}
In the second line we use that 
\begin{align}
\label{eq: convergence with number operator and weyl operators}
&\norm{ \mathcal{N} W^*(\sqrt{N} \varphi_t) U_K \left( \Psi_{N,t+h} - \Psi_{N,t} \right)}
\nonumber \\
&\quad \leq C  \left( 1 + N \norm{\varphi_t}^2 \right)  \norm{\left( \mathcal{N} + 1 \right) U_K \left( \Psi_{N,t+h} - \Psi_{N,t} \right)}
\rightarrow 0 
\end{align}
as $h \rightarrow 0$ because of \eqref{eq: convergence with number operator} and obtain
\begin{align}
&\lim_{h \rightarrow 0}  h^{-1} \left( \beta_K^b(\Psi_{N,t+h}, \varphi_t) -  \beta_K^b(\Psi_{N,t}, \varphi_t) \right)
\nonumber \\
&\quad = -2  N^{-1} \Re \scp{W(\sqrt{N} \varphi_t) \mathcal{N} W^*(\sqrt{N} \varphi_t) U_K \Psi_{N,t}}{i H_{N,K}^{\rm{G}}  U_K \Psi_{N,t} }
\end{align}
with Stone's theorem. Next, we consider
\begin{subequations}
\begin{align}
\label{eq: rigirous derivative beta-b  B 1}
\eqref{eq: rigirous derivative beta-b B}
&= - 2 \Re \int d^3k \,
\scp{\left( \varphi_{t+h}(k) - \varphi_t(k) \right) U_K \Psi_{N,t+h}}{\left( N^{-1/2} a_k - \varphi_t(k) \right) U_K \Psi_{N,t+h}}
\\
\label{eq: rigirous derivative beta-b B 2}
&\qquad 
+ 2  \Re \int d^3k \,
\scp{\left( \varphi_{t+h}(k) - \varphi_t(k) \right) U_K \Psi_{N,t}}{\left( N^{-1/2} a_k - \varphi_t(k) \right) U_K \Psi_{N,t}}
\\
&\qquad +  \norm{\varphi_{t+h} - \varphi_t}^2
\\
&\qquad -  2 \Re \int d^3k \, 
\scp{\left( \varphi_{t+h}(k) - \varphi_t(k) \right) U_K \Psi_{N,t}}{\left( N^{-1/2} a_k - \varphi_t(k) \right) U_K \Psi_{N,t}}.
\end{align}
\end{subequations}
Using the Cauchy--Schwarz inequality we estimate the first two terms by
\begin{align}
&\abs{\eqref{eq: rigirous derivative beta-b B 1} + \eqref{eq: rigirous derivative beta-b B 2}}
\nonumber \\
&\quad \leq 2 \int d^3k \, \abs{\varphi_{t+h}(k) - \varphi_t(k)}
\bigg( \abs{\scp{U_K \Psi_{N,t+h}}{\left( N^{-1/2} - \varphi_t(k) \right) U_K \left( \Psi_{N,t+h} - \Psi_{N,t} \right)}}
\nonumber \\
&\qquad  + \abs{\scp{U_K \left( \Psi_{N,t+h} - \Psi_{N,t} \right)}{\left( N^{-1/2} a_k - \varphi_t(K) \right) U_K \Psi_{N,t}}} \bigg)
\nonumber \\
&\quad \leq 2 N^{-1/2} \norm{\left( \varphi_{t+h} - \varphi_t \right)} 
\bigg(  \norm{\mathcal{N}^{1/2} W^*(\sqrt{N} \varphi_t) U_K \left( \Psi_{N,t+h} - \Psi_{N,t} \right)}
\nonumber \\
&\qquad +
\norm{U_K \left( \Psi_{N,t+h} - \Psi_{N,t} \right)} \norm{\mathcal{N}^{1/2} W^*(\sqrt{N} \varphi_t) U_K \Psi_{N,t}}
\bigg).
\end{align}
Stone's theorem and \eqref{eq: convergence with number operator and weyl operators} then lead to
\begin{align}
&\lim_{h \rightarrow 0} h^{-1} \left( \beta^b(\Psi_{N,t+h}, \varphi_{t+h})  - \beta_K^b(\Psi_{N,t+h}, \varphi_t) \right)
\nonumber \\
&\quad =   - 2 \Re \int d^3k \, 
\scp{ \left( \partial_t \varphi_t(k) \right) U_K \Psi_{N,t}}{\left( N^{-1/2} a_k - \varphi_t(k) \right) U_K \Psi_{N,t}} .
\end{align}
In combination this shows \eqref{eq: time derivative of beta-b product rule a} and \eqref{eq: time derivative of beta-b product rule b}.\medskip


\section*{Acknowledgments}

Financial support by the European Research Council (ERC) under the European Union's Horizon 2020 research and innovation programme (grant agreement No 694227; N.L and R.S.), 
the SNSF Eccellenza project PCEFP2 181153 (N.L) and the Deutsche Forschungsgemeinschaft (DFG) through the Research Training Group 1838: Spectral Theory and Dynamics of Quantum Systems (D.M.) is gratefully acknowledged. 
N.L. gratefully acknowledges support from the NCCR SwissMAP and would like to thank Simone Rademacher and Benjamin Schlein for interesting discussions about the time-evolution of the polaron at strong coupling. D.M. thanks Marcel Griesemer and Andreas W\"unsch for extensive discussions about the Fr\"ohlich polaron.


{}


\begin{thebibliography}{11}

\addcontentsline{toc}{chapter}{Bibliography}


\bibitem{ammarifalconi}
Z.~Ammari and M.~Falconi,
Bohr's correspondence principle for the renormalized Nelson model.
\emph{SIAM J. Math. Anal.} 49(6), 5031--5095 (2017).

\bibitem{benedikterschlein}
N.~{Benedikter}, G.~{de Oliviera}, and B.~{Schlein},
Quantitative derivation of the Gross-Pitaevskii equation.
\emph{Comm. Pure Appl. Math.} 68(8), 1399--1482 (2015).


\bibitem{brenneckeschlein}
C.~Brennecke and B.~Schlein,
Gross-Pitaevskii dynamics for Bose--Einstein condensates.
\emph{Anal. PDE} 12(6), 1513--1596 (2019).

\bibitem{CCFO}
R.~Carlone, M.~Correggi, M.~Falconi  and M.~Olivieri,
Microscopic Derivation of Time-dependent Point Interactions.
\emph{Preprint}, \href{https://arxiv.org/abs/1904.11012}{arXiv:1904.11012} (2019). 

\bibitem{CFO}
M.~Correggi, M.~Falconi  and M.~Olivieri,
Quasi-Classical Dynamics.
\emph{Preprint}, \href{https://arxiv.org/abs/1909.13313}{arXiv:1909.13313} (2019). 


\bibitem{davies}
E.\,B.~Davies,
Particle-boson interactions and the weak coupling limit.
\emph{J. Math. Phys.} 20, 345--351 (1979).



\bibitem{falconi}
M.~Falconi,
Classical limit of the Nelson model with cutoff.
\emph{J. Math. Phys.} 54(1), 012303 (2013).


\bibitem{frankgang}
R.\,L.~Frank and Z.~Gang,
Derivation of an effective evolution equation for a strongly coupled polaron.
\emph{Anal. PDE} 10(2), 379--422 (2017).

\bibitem{frankgang2}
R.\,L.~Frank and Z.~Gang,
A non-linear adiabatic theorem for the one-dimensional Landau--Pekar equations.
\emph{J. Funct. Anal.} 279(7), 108631  (2020).


\bibitem{frankschlein}
R.\,L.~Frank and B.~Schlein,
Dynamics of a strongly coupled polaron.
\emph{Lett. Math. Phys.} 104, 911--929 (2014).

\bibitem{frankseiringer}
R.\,L.~Frank and R.~Seiringer,
Quantum corrections to the Pekar asymptotics of a strongly coupled polaron.
\emph{Preprint}, \href{https://arxiv.org/abs/1902.02489}{arXiv:1902.02489} (2019), Comm. Pure Appl. Math (in press).


\bibitem{ginibrenironivelo}
J.~Ginibre, F.~Nironi, and G.~Velo,
Partially classical limit of the Nelson model.
\emph{Ann. H. Poincar\'e} 7, 21--43 (2006).

\bibitem{ginibrevelo}
J.~Ginibre and G.~Velo,
The classical field limit of scattering theory for nonrelativistic many-boson systems I and II.
\emph{Commun. Math. Phys.} 66(1), 37--76 (1979) and 68(1), 45--68 (1979).

\bibitem{griesemer}
M.~Griesemer,
On the dynamics of polarons in the strong-coupling limit.
\emph{Rev. Math. Phys.} 29(10), 1750030 (2017).
 
\bibitem{griesemerwuensch}
M.~Griesemer and A.~W\"unsch,
Self-adjointness and domain of the Fr\"ohlich Hamiltonian.
\emph{J. Math. Phys.} 57(10), 021902 (2016).
 
\bibitem{gross}
E.\,P.~Gross,
Particle-like solutions in field theory,
\emph{Ann. Phys.} 19, 219--233 (1962).
 

\bibitem{hepp}
K.~Hepp, 
The classical limit for quantum mechanical correlation functions.
\emph{Commun. Math. Phys.} 35, 265--277 (1974). 

\bibitem{hiroshima}
F.~Hiroshima,
Weak coupling limit with a removal of an ultraviolet cutoff for a Hamiltonian of particles interacting with a massive scalar field.
\emph{Infin. Dimens. Anal. Qu.} 1, 407--423 (1998).

\bibitem{jeblickleopoldpickl}
M.~Jeblick, N.~Leopold, and P.~Pickl,
Derivation of the Time Dependent Gross-Pitaevskii Equation in Two Dimensions.
\emph{Commun. Math. Phys.}, 372, 1--69 (2019).  


 
\bibitem{landaupekar}
L.\,D.~Landau and S.\,I~Pekar,
Effective mass of a polaron.
\emph{Zh. Eksp. Teor. Fiz.} 18(5), 419--423 (1948).

\bibitem{leopoldpetrat}
N.~Leopold and S.~Petrat,
Mean-Field Dynamics for the Nelson Model with Fermions.
\emph{Ann. H. Poincar\'e} 20(10), 3471--3508 (2019).


\bibitem{leopoldpickl}
N.~Leopold and P.~Pickl,
Derivation of the Maxwell-Schr\"odinger equations from the Pauli--Fierz Hamiltonian. 
\emph{SIAM J. Math. Anal.} 52(5), 4900--4936 (2020).


\bibitem{leopoldpickl2}
N.~Leopold and P.~Pickl,
Mean-field limits of particles in interaction with quantized radiation fields.
In: D.~Cadamuro, M.~Duell, W.~Dybalski, and S.~Simonella (eds) \emph{Macroscopic Limits of Quantum Systems}, volume 270 of Springer Proceedings in Mathematics \& Statistics, 185--214 (2018).

\bibitem{LRSS}
N.~Leopold, S.~Rademacher, B.~Schlein and R.~Seiringer,
The Landau--Pekar equations: Adiabatic theorem and accuracy.
\emph{Preprint}, \href{https://arxiv.org/abs/1904.12532}{	arXiv:1904.12532} (2019), Anal. \& PDE (in press).


\bibitem{liebyamazaki}
E.\,H.~Lieb and K.~Yamazaki,
Ground-state Energy and Effective Mass of the Polaron.
\emph{Phys. Rev.} 111, 728--733 (1958).

\bibitem{mitrouskaspetratpickl}
D.~Mitrouskas, S.~Petrat and P.~Pickl,  
Bogoliubov corrections and trace norm convergence for the Hartree dynamics.
\emph{Rev. Math. Phys.} 31(8) (2019). 


\bibitem{nelson}
E.~Nelson,
Interaction of nonrelativistic particles with a quantized scalar field.
\emph{J. Math. Phys.} 5(9), 1190--1197 (1964).


\bibitem{pickl1}
P.~Pickl,
A simple derivation of mean field limits for quantum systems.
\emph{Lett. Math. Phys.} 97, 151--164 (2011).

\bibitem{pickl2}
P.~Pickl,
Derivation of the time dependent Gross-Pitaevskii equation with external fields.
\emph{Rev. Math. Phys.} 27(1), 1550003, 45 pp. (2015).

\bibitem{rodnianskischlein}
I.~Rodnianski, B.~Schlein,
Quantum fluctuations and rate of convergence towards mean field dynamics.
\emph{Commun. Math. Phys.} 291(1), 31--61 (2009).

 
\bibitem{teufel}
S.~Teufel, 
Effective $N$-body dynamics for the massless Nelson model and adiabatic decoupling without spectral gap.
\emph{Ann. H. Poincar\'e} 3, 939--965 (2002).
\end{thebibliography}
\end{document}